\documentclass[11pt]{article}
\usepackage[english]{babel}
\usepackage{lmodern}
\usepackage[T1]{fontenc}
\usepackage[utf8]{inputenc}
\usepackage[a4paper]{geometry}
\usepackage[text,italic]{esdiff}
\usepackage{mathtools}
\usepackage{amsfonts}
\usepackage{bbm} 
\usepackage{mathrsfs}
\usepackage{amsthm}
\usepackage{amsmath}
\usepackage{amssymb}
\usepackage{enumitem} 
\usepackage{fancyhdr} 
\usepackage{graphicx}
\usepackage{subcaption}
\usepackage{epstopdf}

\usepackage{wrapfig}

\usepackage{tikz}
\usetikzlibrary{decorations.markings}
\tikzstyle{vertex}=[circle, draw, inner sep=0pt, minimum size=3pt, fill]

\usepackage{thmtools}
\declaretheorem[
	style=remark,
	qed=$\triangle$,
	name=Example,
	numberwithin=section]{example}

\declaretheorem[
	style=remark,
	qed=$\triangle$,
	name=Example: Simple Loop]{loopexample}

\theoremstyle{plain}
\newtheorem{thm}{Theorem}[section]
\newtheorem{prop}[thm]{Proposition}
\newtheorem{lemma}[thm]{Lemma}
\newtheorem{corol}[thm]{Corollary}
\theoremstyle{definition}
\newtheorem{defn}[thm]{Definition}

\usepackage[unicode=true,
bookmarks=true,bookmarksopen=false,
breaklinks=false,pdfborder={0 0 0},colorlinks=false]
{hyperref}
\makeatletter
\addto\extrasenglish{%
	\renewcommand{\itemautorefname}{\@gobble}
	\def\equationautorefname~#1\null{Equation~(#1)\null}
}
\makeatother

\numberwithin{equation}{section}

\def\l{\ell}

\def\N{\mathbb{N}}
\def\Q{\mathbb{Q}}
\def\R{\mathbb{R}}
\def\C{\mathbb{C}}
\def\DD{\mathcal{D}}
\def\V{\mathscr{V}}
\def\H{\mathcal{H}}
\def\L{\mathcal{L}}

\def\U{\mathcal{U}}
\def\ones{\mathbbm{1}}

\DeclareMathOperator{\dom}{dom}
\DeclareMathOperator{\range}{range}
\DeclareMathOperator{\diag}{diag}

\newcommand{\abs}[1]{\left| #1 \right|}
\newcommand{\norm}[1]{\left\Vert #1 \right\Vert}

\newcommand{\carrow}[1]{\overset{\curvearrowright}{#1}}
\newcommand{\veval}[1]{\underline{#1}}
\newcommand{\oveval}[1]{\carrow{\underline{#1}}}
\def\D{\diff}

\newcommand{\Dt}[1][]{\D[#1]{}{t}}
\newcommand{\Dx}[1][]{\D[#1]{}{x}}

\let\colon\relax
\DeclareMathSymbol{:}{\mathpunct}{operators}{"3A}
\DeclareMathSymbol{\colon}{\mathrel}{operators}{"3A}
\renewcommand{\mid}{\colon}

\definecolor{azulUC3M}{RGB}{0,0,102}
\renewcommand*{\title}{Quantum Control at the Boundary: \\ an application to quantum circuits}
\renewcommand*{\author}{Á. Aitor Balmaseda Martín}
\newcommand{\advisor}{Alberto Ibort Latre \\ Juan Manuel Pérez Pardo}
\renewcommand*{\date}{September, 2018}

\begin{document}
	\begin{titlepage}
	\begin{figure}[!t]
		\begin{center}
			\includegraphics[width=15cm]{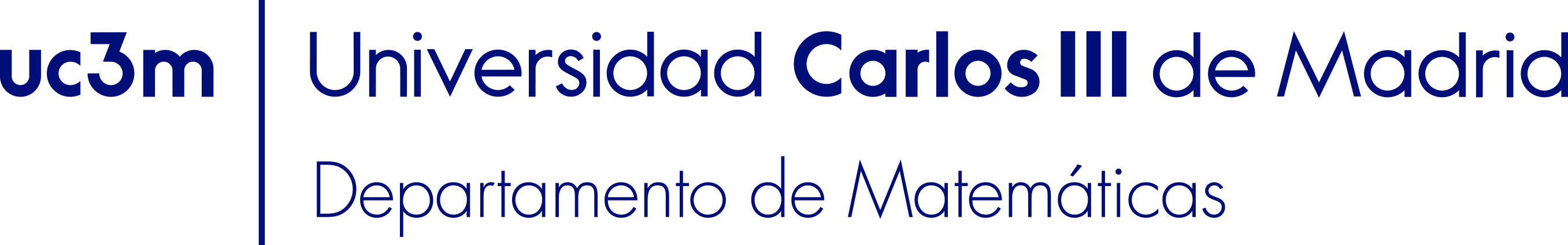}
		\end{center}
	\end{figure}
	\begin{center}
		\color{azulUC3M}
		\vspace{0.4444in}
		\vspace*{0.222in}
		\rule{130mm}{0.5mm}\\
		\vspace*{0.2in}
		\begin{LARGE}
			\textsc{\title} \\
		\end{LARGE}
		\vspace*{0.2in}
		\rule{130mm}{0.5mm}\\
		\vspace*{0.5in}

		\begin{large}
			Master thesis by \\
		\end{large}
		\begin{Large}
			\author \\
		\end{Large}

		\vspace*{0.3in}
		\rule{40mm}{0.1mm}\\
		\vspace*{0.15in}
		\begin{large}
			Work Submitted in Partial Fulfillment of the\\
			Requirements for the Master Degree of \\
			Mathematical Engineering\\
		\end{large}
		\vspace*{0.1in}
		\rule{40mm}{0.1mm}\\
		\vspace*{0.3in}
		\begin{large}
		Supervised by:\\
		\advisor\\
		\end{large}
		\vspace*{1.3in}
		\begin{Large}
			\date
		\end{Large}
	\end{center}
\end{titlepage}

	\thispagestyle{empty}
	\tableofcontents
	\newpage
	\setcounter{page}{1}
	\section{Introduction}
		Approximate controllability for a quantum system on a graph using as control parameters boundary conditions will be proven. This establishes a first theoretical proof of the feasibility of the quantum control at the boundary paradigm.

		It is difficult to express the long-term vision embracing the research aimed to develop ideas around quantum information better than in the opening of the recently launched ``Quantum Manifest'':\footnote{\url{http://qurope.eu/system/files/u567/Quantum\%20Manifesto.pdf}}

		\begin{quote}
			With quantum theory now fully established, we are required to look at the world in a fundamentally new way: objects can be in different states at the same time (superposition) and can be deeply connected without any direct physical interaction (entanglement)... The developments in the leading areas of quantum technologies: communications, simulators, sensors and computers, can be expected to produce transformative applications with real practical impact on ordinary people. The technology tracks showing the underlying scientific and engineering milestones paving the way for disruptive applications are based on predictions from leading scientists in Europe.
		\end{quote}

		A basic requirement for an effective quantum information processing system, quantum sensor, or simulator is the ability to control the quantum state of the system at the individual level, with much better precision than has been achieved before. There have been many attempts to develop technologies that would allow for such delicate control: from atoms in optical lattices and atom chips to trapped ions and superconducting charge qubits. In particular, in the domain of quantum computation, quantum error correction enables fault-tolerant quantum computation to be performed, provided that each elementary operation meets a certain fidelity threshold. Unfortunately, this puts extremely demanding constraints on the allowable errors. Qubits are encoded in the quantum states of physical systems and even if nano-fabricated quantum qubits permit large-scale integration, they usually suffer from short coherence times due to interactions with their environment. The outstanding challenge is to engineer the environment so that it minimally affects the qubit, but still allows qubit control \cite{kawakami_electrical_2014}. On the other hand, superposition states are naturally very sensitive to the environment, and can therefore be used to make very accurate sensors. However, because of the strict operating conditions, such as low temperature, these devices are difficult to integrate with standard solid state devices to be used in industrial and/or commercial applications.

		In spite of all this, quantum control (optimal or not) of coupled (or standalone) quantum spin systems is becoming more and more relevant because of their experimental implementation achievements. For qubits based on spin, a universal single-qubit gate has been realized by a rotation of the spin by any angle about an arbitrary axis \cite{press_complete_2008} and the control and high-fidelity readout of a nuclear spin qubit was shown in \cite{pla_high-fidelity_2013}. Even more, individual spins, associated with vacancies in the silicon carbide lattice, have also been observed and coherently manipulated \cite{morello_quantum_2015}. In the same vein electrical control of a long-lived spin qubit in a Si/SiGe quantum dot was shown in \cite{malissa_room-temperature_2014}. Let us emphasize that interacting or coupled spin systems are fundamental in quantum computation as a network of interacting and controllable spin qubits can act as a quantum computer. However, because of their magnetic and quantum-mechanical nature, the spin qubits must be controlled and measured using radically different techniques as compared to classical, transistor-based bits. Further developments will aim at measuring and controlling the exchange interaction between pairs of spins, to demonstrate a fully functional 2-qubit quantum logic gate (see for instance the analysis of continuous feedback control in \cite{wiseman_quantum_2012}).

		Geometrical control theory has provided the mathematical background to deal with quantum spin control. Khaneja et al.\ showed how to obtain efficient RF pulse trains for two-spin and three-spin NMR systems by finding sub-Riemannian geodesics on a quotient space of $SU(4)$ \cite{khaneja_sub-riemannian_2002} and the subsequent numerical implementations \cite{khaneja_optimal_2005}. We should also mention \cite{moseley_geometric_2004} for a geometric control study of quantum spin systems and \cite{schulte-herbruggen_optimal_2005} for an optimal control discussion of blocks of quantum algorithms (see also \cite[Chps. 5,6]{dalessandro_introduction_2007} and the recent review of geometric optimal control for quantum systems in NMR by Bonnard  et al.\ \cite{bonnard_review_2012} and references therein). More recently a class of optimal control problems for coupled spin systems using geometrical tools was described in \cite{delgado-tellez_optimal_2016}. A review on the current state-of-the-art on quantum optimal control can be found in \cite{glaser_training_2015}.

		However, geometric control theory and its extension to optimal control problems suffers serious drawbacks when extended to genuine infinite dimensional quantum systems and only a few results on controllability of bilinear systems are known (see for instance Beauchard et al.\ \cite{beauchard_local_2005,beauchard_controllability_2006}, and Chambrion et al.\ \cite{chambrion_controllability_2009} and references therein). We should also mention that alternative to geometric control theory is the use of quadrature operators in control theory (see \cite{weedbrook_gaussian_2012} and references therein). There, the conventional approach is introduction of quadrature operators, and then studying and controlling the dynamics of these quadratures, for example through their Wigner functions (see for instance \cite{agarwal_quantum_2012}, \cite{carlini_time-optimal_2014}, \cite{genoni_optimal_2013}).

		The quantum control at the boundary (QCB) method is a radically different approach to the problem of controlling the state of a qubit. Instead of seeking the control of the quantum state by directly interacting with it using external magnetic or electric fields, the control of the state will be achieved by manipulating the boundary conditions of the system. The spectrum of a quantum system, for instance an electron moving in a box, depends on the boundary conditions imposed on it, either Dirichlet or Neumann in most cases. A modification of such boundary conditions modifies the state of the system allowing for its manipulation and, eventually, its control \cite{ibort_quantum_2010}.

		The QCB paradigm has been used to show how to generate entangled states in composite systems by suitable modifications of the boundary conditions \cite{ibort_boundary_2014}. The relation of QCB and topology change has been explored in \cite{perez-pardo_boundary_2015} and recently used to describe the physical properties of systems with moving walls (\cite{facchi_moving_2016}, \cite{facchi_boundaries_2018}, \cite{facchi_quantum_2018}, \cite{facchi_self-adjoint_2018}, \cite{garnero_quantum_2018}), but in spite of its intrinsic interest some basic issues such as the QCB controllability of simple systems has never been addressed.

		This work is devoted to explore the (approximate) controllability of a quantum circuit, that is, of a free quantum system on a graph, by using QCB. Notice that in addition to the relevance of the problem in the context discussed already, quantum systems on graphs provide a setting for universal computation \cite{childs_universal_2009}, thus QCB opens a new way to implement quantum computation.

		In developing the theory it will be shown first, by means of a suitable chosen time-dependent unitary transformation, that the variation of the boundary conditions of the system can be implemented as a time-dependent family of Hamiltonian operators, an idea that was already anticipated in \cite{perez-pardo_boundary_2015}. The particular instance of quasi-periodic boundary conditions on graphs will be worked out explicitly and it will be shown that the system reduces to a bilinear system similar to those studied by Chambrion et al.\ \cite{chambrion_controllability_2009}.

		This work is organized as follows. In Section \ref{sec:preliminaries} we will review some basic aspect both of Hilbert spaces theory and control theory in order to fix the notation we will use. Section \ref{sec:quantum-graphs} introduces some basic ideas of graph theory before reviewing some quantum graphs theory, providing the language we are going to use to describe the spatial setting of the systems studied in this work. On the other hand, Section \ref{sec:magnetic-laplacian} is devoted to the study of magnetic Laplacians, establishing the relation between them and quasi-periodic boundary conditions (which correspond with what we call quasi-$\delta$-type vertex conditions) that will be key to prove controllability using QCB. After introducing this basic concepts related with the spatial description of the system, Section \ref{sec:time-evolution} covers some results on time-dependent Hamiltonians and their associated evolution which will be needed in Section \ref{sec:boundary-driven-dynamics} to prove the (approximate) controllability of some quantum circuit by using QCB.

	\section{Preliminaries} \label{sec:preliminaries}
		\subsection{Hilbert space notions}\label{subsec:preliminaries-hilbert}
			Since our aim is treating quantum mechanical problems, it is clear that Hilbert spaces play a central role throughout the text, as do self-adjoint operators. In order to remind the main concepts we use in the following sections, this section is devoted to summarize basic ideas from Hilbert spaces theory while fixing the notation we use.

			Through this work we will use the letter $\H$ to denote Hilbert spaces, that is, vector spaces endowed with an inner product $\langle \cdot, \cdot \rangle$. We denote by $\L(\H)$ the space of linear (possibly unbounded) operators from $\H$ into itself.

			Let $T \in \L(\H)$ be densely defined. We define its adjoint as the linear operator $T^\dagger \in \L(\H)$ with domain
			\[
				\dom T^\dagger = \{\varphi \in \H \mid \exists \eta \in \H,
				\forall \psi \in \dom T,
				\langle \varphi, T\psi \rangle = \langle \eta, \psi \rangle\}
			\]
			and for every $\varphi \in \dom T^\dagger$ we define $T^\dagger \varphi = \eta$. By the Riesz Lemma, $\varphi$ being in $\dom T^\dagger$ is equivalent to have $\abs{\langle \varphi, T\psi \rangle} \leq C \norm{\psi}$ for all $\psi \in \dom T$.

			Given two operators $T_1, T \in \L(\H)$, $T_1$ is said to be an extension of $T$ if $\dom T_1 \supset \dom T$ and $T_1 \varphi = T \varphi$ for all $\varphi \in \dom T$; we denote it by $T_1 \supset T$.

			We say that $T \in \L(\H)$ is closed if so is its graph,
			\[ \Gamma(T) = \{(\psi, T\psi) \in \H \times \H \mid \psi \in \dom T \},\]
			and we say that it is closable if there exists an extension of $T$ which is closed. It is easy to see that adjoint operators are always closed \cite[Thm. VIII.1]{reed_methods_2012}. For a closed operator $T$ we define its resolvent set, $\rho(T)$, made of the complex numbers $\lambda$ such that $\lambda I - T$ is a bijection of $\dom T$ onto $\H$ with bounded inverse; and for any $\lambda \in \rho(T)$ we define the resolvent of $T$ at $\lambda$, $R_\lambda(T) = (\lambda I - T)^{-1}$. A complex number $\lambda$ is in the spectrum of $T$, $\sigma(T)$, if it is not in $\rho(T)$. The subset of $\sigma(T)$ made of those values $\lambda \in \C$ such that $\lambda I - T$ is not injective (that is, there exists $\H \ni \psi \neq 0$ such that $T\psi = \lambda\psi$) is called the point spectrum of $T$. In this case, $\lambda$ is said to be an eigenvalue and $\psi$ an associated eigenvector (or eigenstate). If $\lambda I - T$ is injective but not surjective with $\range(\lambda I - T)$ dense, $\lambda$ is said to be in the continuous spectrum of $T$, while in the case where $\lambda I - T$ is injective but $\range(\lambda I - T)$ is not dense in $\H$, $\lambda$ is said to be in the residual spectrum of $T$. The essential spectrum of $T$ is the union of both the continuous and residual spectrum.

			We say that $T \in \L(\H)$ is a symmetric operator if $T^\dagger \supset T$, and we say that $T$ is self-adjoint if $T = T^\dagger$ (notice that this implicitly requires that $\dom T = \dom T^\dagger$. It can be proved that every symmetric operator is closable and every self-adjoint operator is closed, \cite[\S VIII.2]{reed_methods_2012}.

			In order to study the evolution of quantum systems one needs to address the problem of finding self-adjoint extensions of a symmetric operator which fits the physical problem under study. For doing that, Von Neumann developed a general theory of self-adjoint extensions of symmetric operators in Hilbert spaces \cite{neumann_allgemeine_1930}. An approach to the problem using Von Neumann's Theorem can be found, for example, in \cite[\S X.1]{reed_methods_1975} but in this work we are going to follow the approach introduced by M. Asorey, A. Ibort and G. Marmo \cite{asorey2005global,asorey2015topology} for finding self-adjoint extensions of first and second order elliptic differential operators which relates boundary conditions and self-adjoint extensions. Other references on the relation of boundary conditions and self-adjoint extensions are the works by G. Grubb \cite{Grubb1968}, A.N. Kochubei \cite{Kochubei1975} and J. Brüning et al.\ \cite{bruning_spectra_2008}. An illustration of the procedure for the one-dimensional case can be found in \autoref{subsec:selfadjoint-extensions}.

			Finally, we will deal with time-dependent Hamiltonians and for doing so we will use notions of continuity and differentiability for operator-valued functions. We will use those given by the so-called strong (operator) topology. The strong topology can be defined in the more general context of convex spaces (see \cite[\S V.7]{reed_methods_2012}) and it is the coarsest topology such that the evaluation maps $T \in \L(\H) \mapsto T \psi \in \H$ are continuous in $\H$ for any fixed $\psi$. Therefore, the notions of continuity and differentiability we are going to use can be stated as follows.

			\begin{defn}\label{def:strong-continuity-differentiability}
				Let $T: \R \to \L(\H)$ be a real, operator-valued function such  that $\DD \coloneqq \dom T(t)$ does not depend on $t$.\footnote{We assumed that $\dom T(t)$ does not depends on $t$ just for simplicity. For a more general definition one could take $\DD = \bigcap_t \dom T(t)$.}
				\begin{enumerate}[label=\textit{(\roman*)},nosep]
					\item We say that $T$ is strongly continuous (or continuous in the strong sense) if it is continuous with respect to the strong operator topology; that is, if the mapping $t \mapsto T(t) \psi$ is continuous (in the sense of the Hilbert space) for any $\psi \in \DD$.
					\item We say that $T$ is strongly differentiable if the mapping $t \mapsto T(t) \psi$ is differentiable (in the sense of the Hilbert space) for every $\psi \in \DD$; that is, $T$ is strongly differentiable if for every $t$ there exists an operator $\D{T}{t}(t) \in \L(\H)$, which we call strong derivative of $T(t)$, with domain $\dom \D{T}{t}(t) \subset \mathcal{D}$ and defined by the limit
					\[
						\lim_{h \to 0} \frac{\norm{T(t + h)\psi - T(t)\psi - h \D{T}{t}(t) \psi}}{\abs{h}} = 0, \qquad \text{for every } \psi \in \dom \D{T}{t}(t).
					\]
				\end{enumerate}
			\end{defn}

		\subsection{Controllability notions}\label{subsec:controllability-notions}
			Since this work intends to discuss controllability of some quantum systems, it is convenient to introduce briefly the notions of controllability we use. From a mathematical point of view, a quantum system is just a dynamical system on a vector space which evolves linearly. Thus, the ideas in classical control can be applied straightforwardly. This approach has been successful for finite-dimensional quantum systems (or finite-dimensional approximations to infinite-dimensional ones), leading to complete characterizations of the various notions of controllability of classical systems in the (finite-dimensional) quantum case (see, for instance, \cite{dalessandro_introduction_2007} and references therein). However, it is not straightforward to extend that approach to infinite-dimensional quantum systems for several reasons such as the analytical difficulties raising from the appearance of unbounded operators.

			Here we will address the problem of state controllability of some quantum systems, defined as follows:
			\begin{defn}
				Let $\psi(t;u)$ denote the evolution of a quantum system with controls $u$ starting in some initial state $\psi_0$; we say that such a quantum system is (state) controllable if for every initial state $\psi_0$ and target state $\psi_T$, there exists a time $T>0$ and a control $u$ such that $\psi(T;u) = \psi_T$.
			\end{defn}
			There are well-known results on state controllability for finite-dimensional quantum systems. For instance, it has been considered from a Lie-theoretical point of view, determining necessary and sufficient conditions for a given dynamical group to act transitively on the set of normalized pure states (see for instance \cite{dalessandro_introduction_2007}). However, these results doesn't generalize naively to the infinite-dimensional case even when they apply finite-dimensional approximations of it. It can be easily checked that the harmonic oscillator control problem is not controllable (see for instance \cite{mirrahimi_controllability_2004}), while its truncation up to the first $n+1$ energy eigenstates is controllable for every $n$ \cite{ramakrishna_controllability_1995}.

			Although there is not a result giving necessary and sufficient conditions on (state) controllability for infinite-dimensional quantum systems, there are some results on controllability of the 1D Schrödinger equation (for example, \cite{beauchard_local_2005,beauchard_controllability_2006}).

			A less restrictive notion is that of approximate (state) controllability:
			\begin{defn}
				Let $\psi(t;u)$ denote the evolution of a quantum system with controls $u$ starting in some initial state $\psi_0$; we say that such a quantum system is approximately (state) controllable if for every initial state $\psi_0$, every target state $\psi_T$ and every $\varepsilon > 0$, there exists a time $T>0$ and a control $u$ such that $\norm{\psi(T;u) - \psi_T} < \varepsilon$.
			\end{defn}
			This is the notion of controllability we consider in this work, because of two main reasons: first, it is suitable for experimental purposes where absolute precision cannot be achieved; second, it is what one should expect in a case where all the finite-dimensional approximations of an infinite-dimensional system are exactly controllable. In this sense, it is noticeable the result by Chambrion et al.\ in \cite{chambrion_controllability_2009}, which is a key piece on the proof of the main result in this work (\autoref{thm:controllability-piecewise-smooth}).

	\section{Quantum graphs and quantum wires}\label{sec:quantum-graphs}
		\subsection{Basic notions of graph theory}\label{subsec:graph-theory}
			A directed graph (sometimes abbreviated as digraph) $G = (V, E, \partial)$ consists on a set of vertices $V = V(G)$, a set of edges $E = E(G)$ and an orientation $\partial = \partial_G : E \to V \times V$ that maps any edge $e\in E$ to the pair of vertices $(\partial_- e, \partial_+ e) \in V \times V$ composed by its starting vertex $\partial_- e$ and the target vertex $\partial_+ e$. In what follows we will assume that the graphs we work with are finite, that is, the vertices and edges sets have a finite number of elements.

			For two subsets of vertices $A, B \subset V$, we denote
			\[
				E^+ (A,B) \coloneqq \{e \in E \mid \partial_- e \in A, \partial_+ e \in B\}.
			\]
			the set of edges starting on a vertex in $A$ which ends on a vertex in $B$. Likewise we denote $E^-(A,B) = E^+(B,A)$ and its disjoint union $E(A,B) = E^+(A,B) \sqcup E^-(A,B)$ whose elements are all the edges between $A$ and $B$ regardless of their orientation. Following the same idea we denote the set of edges between two vertices by $E(v,w)$ and
			\[ E^\pm_v \coloneqq E^\pm(V,v) = \{e\in E \mid \partial_\pm e = v\} \]
			denotes the set of edges starting ($-$) or ending ($+$) on $v$. Finally we will denote by $E_v = E^+_v \sqcup E^-_v$ the set of edges incident to $v$, and by $\deg v \coloneqq |E_v|$ the degree of a vertex. Note that we have by definition $E_v$ contains twice every edge with $\partial_+ e = \partial_- e$, and thus loops appears twice on $E$.

			A concept extensively used in this work is that of vertex spaces.
			\begin{defn}[Definition 2.1.3 in \cite{post2012spectral}]\label{def:vertex-space}
				Let $G = (V, E, \partial)$ be a directed graph.
				\begin{enumerate}[label={\itshape(\roman*)}]
					\item We denote the maximal vertex space at vertex $v\in V$ by $\V_v \coloneqq \C^{E_v}$, that is the vector space $\V_v \ni F(v) = \{F_e(v)\}_{e\in E_v}$ with $\deg v$ complex components (one for each edge incident to $v$).
					\item The corresponding total vertex space associated with $G$ are
					\[
						\V \coloneqq \bigoplus_{v\in V} \V_v.
					\]
					This space is endowed with the norm
					\[ \norm{F}_\V^2 \coloneqq \sum_{v \in V} \sum_{e \in E_v} |F_e(v)|^2. \]
				\end{enumerate}
			\end{defn}
			A general vertex space at the vertex $v$ can be defined as a linear subspace of $\V_v$, and one can define the associated total vertex space analogously; however, we will only use along this work the maximal vertex spaces.

			Until now we have only introduced discrete directed graphs in which the vertices play a major role while edges are understood just as relationships between vertices, but it is turn to introduce metric graphs in which edges play the central role.

			\begin{defn}
				A graph $G$ (directed or not) is a metric graph if associated to every edge $e \in E$ there is a length $\l_e \in (0, \infty)$.
			\end{defn}
			Even if sometimes it might be convenient to allow edges to have infinite length, the aim of this work is to study one dimensional quantum systems with boundary, and so we restrict to compact graphs; i.e. finite metric graphs whose edges have finite length.

			In a metric graph every edge $e\in E$ can be identified with an interval $I_e = [0,\l_e]$ and doing so a coordinate $x_e$ is introduced over that edge. If the graph $G$ is a digraph, that coordinate must be consistent with the edge orientation (meaning $x_e = 0$ at $\partial_- e$ and $x_e = \l_e$ at $\partial_+ e$), while if there weren't an orientation on $G$ we inevitably induce one when choosing the vertex in which $x_e = 0$. When there is no confusion we will denote by $x$ that coordinate, dropping the subscript. Note that by identifying edges with real intervals, $G$ is endowed with a natural metric space structure.\footnote{Here we have introduced a slight abuse of notation which will continue for the rest of this dissertation: we are intentionally denoting by $G$ both the underlying digraph and the metric space associated to it.} The distance between two points in $G$ (i.e. two points in the edges of $G$) is defined as the accumulative length induced by coordinates $x_e$ of the shortest path in $G$ between them.

			The identification between edges and intervals allows us to describe functions over the graph $\psi: G \to \C$ as a family of functions of a real variable $\{\psi_e: I_e \to \C\}_{e\in E}$ labeled by the edges. Therefore, notions of analysis can be brought to $G$, defining the space of square-integrable functions over the graph as
			\[
				L_2(G) \coloneqq \bigoplus_{e \in E} L_2(I_e),
				\quad
				\norm{\psi}_{L_2(G)}^2 \coloneqq \sum_{e\in E} \norm{\psi_e}_{L_2(I_e)}^2 = \sum_{e\in E}\int_{I_e} |\psi_e(x)|^2 \, dx;
			\]
			that is, functions over the graph whose edge restrictions $\{f_e\}_{e\in E}$ are square-integrable. Likewise we define Sobolev type spaces over $G$:
			\[
				\tilde{H}^k(G) \coloneqq \bigoplus_{e \in E} H^k(I_e), \qquad
				\norm{\psi}_{\tilde{H}^k(G)}^2 \coloneqq \sum_{e\in E} \norm{\psi_e}_{H^k(I_e)}^2.
			\]
			Note that functions in this Sobolev type spaces are not assumed to be regular at the vertices in any sense.

			In next sections we will use the vertex evaluation of a function $\psi: G \to \C$, both oriented and unoriented:
			\[
				\veval{\psi} \coloneqq \{\{\veval{\psi}_e(v)\}_{e \in E_v}\}_{v \in V} \in \V, \qquad
				\oveval{\psi} \coloneqq \{\{\oveval{\psi}_e(v)\}_{e \in E_v}\}_{v \in V} \in \V,
			\]
			where
			\[
				\veval{\psi}_e(v) \coloneqq \begin{cases}
					\psi_e(0)		&	\text{if } v = \partial_- e, \\
					\psi_e(\l_e)	&	\text{if } v = \partial_+ e,
				\end{cases}
				\qquad\qquad
				\oveval{\psi}_e(v) \coloneqq \begin{cases}
					-\psi_e(0)		&	\text{if } v = \partial_- e, \\
					\psi_e(\l_e)	&	\text{if } v = \partial_+ e.
				\end{cases}
			\]

			It is easy to check that, using the notation we have introduced, the integration by parts formula for functions $\psi,\varphi \in \tilde{H}^1(G)$ can be written as
			\[ \langle \psi', \varphi \rangle_{L_2(G)} = \langle \veval{\psi}, \oveval{\varphi} \rangle_\V - \langle \psi, \varphi' \rangle_{L_2(G)}, \]
			where $\psi' = \{\psi_e'\}_{e \in E}$ denotes the edge-by-edge derivative of $\psi$, and analogously for $\varphi$.

		\subsection{Quantum graph and vertex conditions}\label{subsec:selfadjoint-extensions}
			Now that we have defined metric graphs, it is natural to define differential operators over them, which lead us to quantum graphs.
			\begin{defn}
				A quantum graph $(G, H)$ consists on a metric graph $G$ and a self-adjoint differential operator $H: \DD \subset \tilde{H}^k(G) \to L_2(G)$ for some $k$.
			\end{defn}

			In this work we will consider two families of differential operators: the standard Laplacian and the magnetic Laplacian. In this subsection, we will use the former to introduce the notion of vertex conditions and in the next section we will cover the properties of magnetic Laplacians. Given a metric graph, the standard Laplacian over it takes $\psi \in \DD \subset L_2(G)$ and maps it into $\Delta\psi = \psi'' = \{\psi_e''\}_{e \in E}$; that is, it acts as the usual Laplacian on every interval $I_e$ associated to each edge. We have yet to determine the domain $\DD$ of $\Delta$, and for $(G, \Delta)$ to be a quantum graph we have to chose $\DD$ such that $\Delta$ is self-adjoint. That is the goal of this subsection: to find the conditions defining $\DD$ for $\Delta$ to be self-adjoint.

			Since we are dealing with a second-order differential operator, it is natural to take $\DD \subset \tilde{H}^2(G)$. Then, by the integration by parts formula, we have
			\[
				\langle \varphi, \Delta \psi \rangle_{L_2(G)}
				- \langle \Delta \varphi, \psi \rangle_{L_2(G)}
				= \langle \veval{\varphi}, \oveval{\psi'} \rangle_\V
				- \langle \veval{\varphi'}, \oveval{\psi} \rangle_\V
				\eqqcolon \Sigma(\psi,\varphi).
			\]
			Therefore, in order to have $\Delta$ self-adjoint $\DD$ must be such that $\dom \Delta^\dagger = \DD$ and hence the boundary term $\Sigma(\psi, \varphi) = 0$ for any $\psi, \varphi \in \DD$. It is clear that $\psi \in \tilde{H}^2_0(G) = \bigoplus_{e \in E} H_0^2(I_e)$ implies $\Sigma(\psi, \varphi) = 0$ for any $\varphi \in \tilde{H}^2(G)$; that is, taking $\dom \Delta = \tilde{H}^2_0(G)$ we have a symmetric operator. In order to find self-adjoint extensions of $\Delta$ we follow the method introduced by M. Asorey, A. Ibort and G. Marmo \cite{asorey2005global,asorey2015topology}. In the case we are studying, the Laplacian over a one-dimensional manifold, the procedure is quite simple since boundary data form a finite dimensional vector space, $\V$. Here we will describe the method without going into the details, for what we refer to the original papers \cite{asorey2005global,asorey2015topology,bruning_spectra_2008,Kochubei1975}.

			As we have defined $\psi'$, we have
			\[
				\oveval{\psi'}_e(v) = \begin{cases}
					-\psi'_e(0)	&	\text{if } v = \partial_- e, \\
					\psi'_e(\l_e)	&	\text{if } v = \partial_+ e;
				\end{cases}
			\]
			i.e. $\oveval{\psi'}_e(v)$ is the normal derivative of $\psi_e$ evaluated on the point of $I_e$ corresponding to the vertex $v$. Therefore, we will denote it by $\veval{\dot{\psi}} = \oveval{\psi'}$ in order to get a notation similar to that used by Asorey et al.\ With this notation we have
			\begin{equation}\label{eq:boundary-terms}
				\Sigma(\psi,\varphi) =
				\langle \veval{\varphi}, \veval{\dot{\psi}} \rangle_\V
				- \langle \veval{\dot{\varphi}}, \veval{\psi} \rangle_\V,
			\end{equation}
			and it is easy to check
			\[
				i\Sigma(\psi, \varphi)
				= \langle \veval{\varphi} - i\veval{\dot{\varphi}}, \veval{\psi} - i\veval{\dot{\psi}} \rangle _\V
				- \langle \veval{\varphi} + i\veval{\dot{\varphi}}, \veval{\psi} + i\veval{\dot{\psi}} \rangle_\V.
			\]
			Hence, the domains of self-adjoint extensions are those maximal subspaces such that $\Sigma$ vanishes on them. It can be proved that those domains are characterized by a unitary operator $U \in \U(\V)$,
			\[
				\dom \Delta_U = \{\psi\in \tilde{H}^2(G) :
				\veval{\psi} - i \veval{\dot{\psi}} = U (\veval{\psi} + i \veval{\dot{\psi}})\}.
			\]
			A complete proof of this fact can be found in \cite{ibort_self-adjoint_2015}, \cite{bruning_spectra_2008} and \cite{Kochubei1975}.

			The condition $\veval{\psi} - i \veval{\dot{\psi}} = U (\veval{\psi} + i \veval{\dot{\psi}})$ enforces some relations between the values at the vertices of $\psi$ and its derivative, and that is why this kind of conditions are called vertex conditions. Note that in the most general case the conditions on $\psi$ enforces relations between values of the function and its derivative at different vertices; i.e. the condition $\veval{\psi} - i \veval{\dot{\psi}} = U (\veval{\psi} + i \veval{\dot{\psi}})$ is not a local vertex condition but a global one. It can be shown that global vertex conditions can be transformed into local ones by modifying the graph (see \cite[\S1.4.5]{berkolaiko2013introduction}).

			Local vertex conditions are those whose unitary operator $U$ is such that it does not relate values of $\psi$ and its derivative at different vertices; that is, $U$ decomposes as a direct sum of unitary operators over $\V_v$:
			\[ U = \bigoplus_{v\in V} U_v, \qquad U_v \in \U(\V_v). \]
			Therefore, local vertex conditions can be expressed as
			\[
				\veval{\psi}(v) - i \veval{\dot{\psi}}(v) = U_v [\veval{\psi}(v) + i \veval{\dot{\psi}}(v)],
				\qquad \forall v \in V,
			\]
			where $\veval{\psi}(v) = \{\veval{\psi}_e(v)\}_{e \in E_v} \in \V_v$ for any function $\psi$ on $G$. Local vertex conditions that make the Laplacian on $G$ a self-adjoint operator can be written in other equivalent ways, and a more complete characterization of them can be found in \cite[Thm. 1.4.4]{berkolaiko2013introduction}.

			Finally, we conclude the subsection with a standard example of local vertex conditions: the \emph{Kirchhoff-Neumann} vertex conditions. A function $\psi$ satisfies Kirchhoff-Neumann vertex conditions if it is continuous at every vertex $v \in V$ and it holds $\sum_{e \in E_v} \veval{\dot{\psi}}_e(v) = 0$ for all vertices. Notice that, since the domains that we consider are always a subset of $\tilde{H}^1(G)$, $\psi$ is guaranteed to be continuous except, at most, at the vertices due to Sobolev embeddings (cf. \cite{adams_sobolev_2003}). Therefore, Kirchhoff-Neumann conditions require $\psi \in \dom \Delta$ to be continuous on $G$. In order to check that this vertex condition defines a self-adjoint extension of the Laplacian we should, for example, show that they are given by a unitary $U$ as described above; however, we will have to wait until the next subsection to introduce the ideas used to obtain such unitary operator. For now, let us just check that the boundary term $\Sigma(\psi,\varphi)$ vanishes. For checking that, it suffices to see that for $\psi, \varphi$ satisfying Kirchhoff-Neumann vertex conditions it holds
			\[
				\langle \veval{\psi}, \veval{\dot{\varphi}} \rangle_\V
				= \sum_{v \in V} \sum_{e \in E_v} \overline{\veval{\psi}_e(v)} \veval{\dot{\varphi}}_e(v)
				= \sum_{v \in V} \overline{\psi(v)} \sum_{e \in E_v} \veval{\dot{\varphi}}_e(v) = 0,
			\]
			where we have used continuity of $\psi$ at $v$. Substituting into \autoref{eq:boundary-terms} one gets $\Sigma(\psi, \varphi) = 0$.

		\subsection{Quantum wires}\label{subsec:quantum-wires}
			In the previous subsection we have reviewed how self-adjoint extensions of the Laplacian are parametrized by unitary operators, so that given a unitary operator in $\U(\V)$ one has the associated self-adjoint extension. But one is often interested on finding all the self-adjoint extensions such that functions on their domains satisfy certain (vertex) conditions, and this subsection will provide the ideas needed for solving that kind of problems. Before we introduce quantum wires, the particular problem we are interested in, let us write vertex conditions in a more convenient way. Considering vertex conditions given by $U \in \U(\V)$, let $P_1$ and $P_{-1}$ be the orthogonal projectors over the eigenspaces of $U$ associated with eigenvalues $1$ and $-1$ respectively, and denote by $P^\perp = I - P_1 - P_{-1}$ the orthogonal projection onto the space spanned by the eigenvectors of $U$ associated to other eigenvalues. Note that, since $\V$ is a linear space of finite dimension, $U$ has no essential spectrum. Multiplying the vertex condition
			\[ \veval{\psi} - i \veval{\dot{\psi}} = U(\veval{\psi} + i\veval{\dot{\psi}}) \]
			by $P_{-1}$ and using $P_{-1} U = - P_{-1}$ it follows
			\[ 2P_{-1} \veval{\psi} = 0. \]
			Multiplying by $P_1$ and applying $P_1 U = P_1$ one gets
			\[ 2iP_1 \veval{\dot{\psi}} = 0. \]
			Finally, multiplying the vertex condition by $P^\perp$ it is easy to find
			\[ P^\perp \veval{\dot{\psi}} = -i (I + U)^{-1}_\perp(I - U) P^\perp \veval{\psi}, \]
			where $(I + U)_\perp = (I + U) P^\perp$ is the restriction of $(I+U)$ to the subspace $P^\perp \V$ and thus it is invertible since $-1 \notin \sigma(P^\perp U)$. With the notation introduced, vertex conditions can be expressed as
			\[ \left\{ \begin{alignedat}{2}
				&P_{-1} \veval{\psi} = 0,	\\
				&P_1 \veval{\dot{\psi}} = 0, \\
				&P^\perp \veval{\dot{\psi}} = -i (I + U)^{-1}_\perp(I - U) P^\perp \veval{\psi}.
			\end{alignedat} \right. \]
			Therefore, it is clear that conditions on the vertex values of functions fix the eigenspace of $U$ associated with -1, while conditions on the vertex values of normal derivatives fix the eigenspace of $U$ associated with 1 and conditions that mix vertex values of functions and their normal derivatives fix the rest of $U$.

			With this description of the vertex conditions we can \emph{easily} find self-adjoint extensions whose vertex conditions impose some relations to the vertex values of functions (and their derivatives) in its domain, and we are able now to study the announced particular case: quantum wires. Quantum graphs are appropriate models to describe the dynamics of one-dimensional quantum systems, and so they are appropriate to study the dynamics of a free\footnote{In this section, we have focused on free particles (i.e. $H = -\Delta$) for simplicity, but all the reasoning applies to more general Hamiltonians describing particles moving under the effect of a potential.} particle moving on a family of intervals $\{I_e\}_{e=1}^m$, each one of the form $I_e = [a_e, b_e]$. It is obvious that $G = \bigcup_e I_e$ is a metric graph whose vertices are extremes of each $I_e$ and edges between them are identified with the intervals. Thus, the state of the quantum particle can be described by a wave function $\psi \in L_2(G)$ and the theory reviewed in the previous subsection can be applied in order to find self-adjoint extensions of the Hamiltonian $H = -\Delta$.

			Quantum wires are the analogous of classic wires, one-dimensional spaces on which a particle can move, and therefore the dynamics of such a particle can be represented by a quantum graph. But now, for a given graph $G$ representing a quantum wire, not all self-adjoint extensions of the Laplacian are allowed, since it is natural to ask for continuity of the wave function on $G$, including its vertices. Also, only local vertex conditions are allowed since a non-local vertex condition would imply that a probability current appears between edges which are not connected. In order to find the class of all self-adjoint extensions of the Laplacian whose vertex conditions are both local and compatible with continuity of $\psi$, we will use the projections introduced before. Using the ideas introduced at the beginning of this subsection, for every $v$ denote now by $P_v$ the projector onto $\ker(U_v + I)$, with $U_v$ the unitary operators defining the local vertex conditions which enforce continuity of the wave functions. Since continuity only implies vertex values of $\psi$, $P_v$ must be such that $P_v \veval{\psi}(v) = 0$ is satisfied if and only if $\psi$ is continuous at $v$; in other words, $P_v \veval{\psi}(v) = 0$ must be satisfied if and only if $\veval{\psi}(v)$ is a uniform vector, that is,
			\[
				\veval{\psi}(v) = c \veval{1}(v)
			\]
			for some $c \in \mathbb{R}$, where $\veval{1}(v) = \{1\}_{e \in E_v}$ represents the uniform vector on $\V_v$ with ones in all its entries. Therefore $P_v$ is the orthogonal projector whose kernel consists on uniform vectors, and it is then clear that the orthogonal projector over $\ker(U_v + I)^\perp$, $P_v^\perp = I - P_v$, must be
			\[
				P_v^\perp = \frac{\veval{1}(v) \veval{1}(v)^\dagger}{\norm{\veval{1}(v)}_{\V_v}} = (\deg v)^{-1} \ones,
			\]
			with $\ones$ the square matrix of size $\deg v$ whose entries are ones.\footnote{Note we have simplified the notation introduced above for the projectors: here $P_v$ corresponds with $P_{-1}$ above and $P_v^\perp = P_1 + P^\perp$ on the vertex $v$.} Note that $P^\perp_v$ is a rank-one projector, and therefore $U_v$ has to be of the form
			\[ U_v = e^{i\delta} P_v^\perp - (I - P^\perp_v) = \frac{e^{i\delta} + 1}{\deg v} \ones - I \]
			where $e^{i\delta}$, with $\delta \in \R$, is the eigenvalue of $U_v$ associated to the projector $P_v^\perp$.

			Using $U_v = e^{i\delta} P_v^\perp - P_v$ and the fact that $\veval{\psi}(v) \in \range(P_v^\perp)$ it is easy to see that the vertex conditions associated with $U_v$ are such that
			\[ P_v^\perp \veval{\dot{\psi}}(v) = \tan \frac{\delta}{2} \psi(v). \]
			Substituting into that equation $P_v^\perp = (\deg v)^{-1} \ones$ it follows
			\[ \sum_{e \in E_v} \veval{\dot{\psi}}_e(v) = -\deg v \tan \frac{\delta}{2} \psi(v), \]
			and thus the vertex condition associated to $U_v$ are such that for every $v \in V$,
			\begin{equation}\label{eq:delta-vertex-cond}
				\left\{ \begin{alignedat}{2}
					&\psi \text{ is continouos at } v, \\
					&\sum_{e \in E_v} \veval{\dot{\psi}}_e(v) = -\deg v \tan \frac{\delta}{2} \psi(v).
				\end{alignedat} \right.
			\end{equation}
			This vertex conditions are called $\delta$-type vertex conditions and the particular case $\delta=0$ corresponds to Kirchhoff-Neumann vertex conditions introduced in the previous subsection. Since we have found the unitary operator defining Kirchhoff-Neumann vertex conditions, it is now proven that they define a self-adjoint extension of the Laplacian.

			Hence, we have showed that assuming local vertex conditions so that no particle transference occurs between unconnected edges, quantum wires are described by quantum graphs whose vertex conditions are of $\delta$-type.

	\section{Magnetic Laplacian} \label{sec:magnetic-laplacian}
		\subsection{Definition and some properties}
			Another family of Hamiltonian operators that appears often consists on those Hamiltonians which describe quantum systems on which magnetic fields are applied. In one dimensional systems, the appropriate Hamiltonian takes the form
			\[ H = -\left( \frac{d}{dx} - i A(x) \right)^2 \eqqcolon -D^2, \]
			where $A: G \to \R$ stands for the magnetic vector potential (more precisely, its projection onto the direction tangent to the edges of $G$)\footnote{The definition of $D$ must be understood \emph{edge by edge}; that is, for $\varphi = \{\varphi_e\}_{e\in E(G)}$, we have $D \varphi = \{\frac{d \varphi_e}{dx_e} - iA_e(x)\varphi_e\}_{e \in E(G)}$} and so we will consider that the restriction of $A$ to any edge is continuous (we will call a function over $G$ satisfying this \emph{edge-continuous}). From its definition it is clear that the role previously played by $\Delta$ is now played by $D^2$, and that is why $D^2$ is called \emph{magnetic Laplacian}.

			If we proceed as we did for the standard Laplacian, integrating by parts it is straightforward to find the \emph{boundary term}
			\[
				i\tilde{\Sigma}(\varphi,\psi) \coloneqq i( \langle \psi, D^2 \varphi \rangle_{L_2(G)} - \langle D^2 \psi, \varphi \rangle_{L_2(G)})
				= \langle \veval{\psi} + i\oveval{D \psi}, \veval{\varphi} + i\oveval{D \varphi} \rangle_\V
				- \langle \veval{\psi} - i\oveval{D \psi}, \veval{\varphi} - i\oveval{D \varphi} \rangle _\V.
			\]
			Therefore, analogously to the case of the standard Laplacian, self-adjoint extensions of $D^2$ are parametrized by an unitary operator $U \in \U(\V)$ with
			\[
				\dom D^2_U= \{\varphi\in \tilde{H}^2(G) :
				\veval{\varphi} - i \oveval{D\varphi} = U (\veval{\varphi} + i \oveval{D\varphi})\},
			\]
			where $\tilde{H}^2(G)$ is the Sobolev space defined on $G$ (see \autoref{subsec:graph-theory}).

			It is going to be convenient for the next section to keep in mind the following well-known property about magnetic Laplacians (see, e.g., \cite{kostrykin2003quantum} for a more detailed study of this properties).

			\begin{prop}\label{prop:magnetic-Laplacian-equiv}
				Let $D^2_U$ be a self-adjoint extension of the magnetic Laplacian associated to a vector potential $A$. Then, for any $\tilde{A}$ there exists a self-adjoint extension of the associated magnetic Laplacian, $\tilde{D}^2_V$, and an isometry $T$ on $L_2(G)$ mapping $\dom \tilde{D}^2_V$ into $\dom D^2_U$ such that
				\[ T^{-1} D^2_U T = \tilde{D}^2_V. \]
				Moreover, $V = \veval{T}^{-1}U\veval{T}$ with $\veval{T}$ the restriction at vertices\footnote{By restriction to the edges of $T$ we mean the operator $\veval{T}: \V \to \V$ such that $\veval{T\vphantom{\varphi}} \,\veval{\varphi} = \veval{T\varphi}$ for any $\varphi$.} of $T$.
			\end{prop}
			\begin{proof}
				As we already said, the magnetic vector potential has to be continuous edge by edge and therefore, by the Poincaré Lemma, there exists $\chi: G \to \R$ differentiable edge by edge such that $\chi' = A - \tilde{A}$. Let $T$ denote the multiplication map defined by
				\begin{equation}\label{eq:def-T}
					T: \{\varphi_e\}_{e \in E} \in L_2(G) \mapsto \{e^{i \chi_e} \varphi_e\}_{e \in E} \in L_2(G).
				\end{equation}
				It follows directly from this definition that $T$ is an isometry on $L_2(G)$. Using the product rule, it is easy to check that
				\[
					\left(\frac{d}{dx} - iA(x)\right) T \psi = T\left(\frac{d}{dx} - i\tilde{A}(x)\right)\psi.
				\]
				Evaluating at the vertices, it follows
				\begin{equation*}\label{eq:prop1.1-1}
					\veval{D_U T \psi} = \veval{T \tilde{D}_V \psi} = \veval{T\vphantom{\psi}}\,\veval{\tilde{D}_V \psi},
				\end{equation*}
				where $\veval{T}$ is the diagonal matrix $\veval{T} = \diag(\{e^{i\underline{\chi}_e(v)}\}_{e \in E, v\in V})$.

				Using this, it is straightforward to show that for any $\varphi \in \dom D^2_U$, $\psi = T^{-1}\varphi$ is in $\dom \tilde{D}^2_V$, if $V = \veval{T}^{-1}U\veval{T}$. Moreover,
				\[ D^2_U T \psi = T \tilde{D}^2_{\veval{T}^{-1}U\veval{T}}, \]
				which concludes the proof.
			\end{proof}

			As a consequence of the previous property, we show the following result which will allow us to consider not only edge-continuous vector potential, but edge-constant ones (i.e., with a constant value on each edge).

			\begin{corol}\label{corol:magnetic-Laplacian-constant}
				Every self-adjoint extension of a magnetic Laplacian $D^2_U$, associated with potential $A$, is equivalent to one associated with a edge-constant potential such that $\dom D^2_U = \dom \tilde{D}^2_U$.
			\end{corol}
			\begin{proof}
				Take $\tilde{A}_e = \l_e^{-1} \int_0^{\l_e} A_e(x) \, dx$ and $\chi_e(x) = \int_0^{x} (A_e(x) - \tilde{A}_e) \, dx$. Define
				\[ \tilde{D}^2_U = \left(\frac{d}{dx} - i \tilde{A}\right)^2 \]
				and $T$ as in \autoref{prop:magnetic-Laplacian-equiv}; it follows that $\veval{T} = I$ and therefore by \autoref{prop:magnetic-Laplacian-equiv}
				\[ T^{-1} D^2_U T = \tilde{D}^2_U. \]
			\end{proof}

			Finally, it follows straightforwardly from \autoref{prop:magnetic-Laplacian-equiv} a result which will be the base for the next subsection.
			\begin{corol}\label{corol:magnetic-standard-equiv}
				For every magnetic Laplacian, $D_U^2$, there is an equivalent self-adjoint extension of the standard Laplacian. Moreover, if $T$ is the multiplication operator defined on Equation \eqref{eq:def-T} with $\chi$ such that $\chi' = A$, then
				\[ T^{-1} D^2_U T = \Delta_{\underline{T}^{-1}U\underline{T}}. \]
			\end{corol}
		\newpage
		\subsection{Relationship with quasi-periodic boundary conditions} \label{subsec:quasi-periodic-BC}
			\begin{wrapfigure}[8]{r}{0.2\textwidth}
				\vspace{-3.5mm}
				\centering
			    \includegraphics[width=0.2\textwidth]{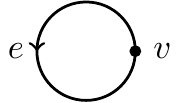}
			    \caption{Simple loop graph, $G_0$.}
			    \label{fig:example0}
			\end{wrapfigure}
			The relationship between standard and magnetic Laplacians established in \autoref{corol:magnetic-standard-equiv} can be used to study easily the generalization of quasi-periodic boundary conditions to the general graph setting. Let us first introduce standard quasi-periodic boundary conditions in \autoref{example:0-1}.

			\begin{loopexample}\label{example:0-1}
				Let us consider the graph $G_0$ in \autoref{fig:example0}, which is the most simple quantum graph with just one vertex and one edge. In what follows, we will come back to this simple case in order to show the techniques we will develop.

				Note that in this case $E_v = E$ and remember that loops (edges starting and ending on the same vertex) appears twice in $E$. For clearness of the notation, we will denote this edges $e_+$ or $e_-$, considering $E_v^\pm = \{e_\pm\}$. This way, when we write $\phi_{e_-}(v)$ (resp. $\phi_{e_+}(v)$) we mean $\phi_e(0)$ (resp. $\phi_e(\l_e)$).

				A function $\phi: G_0 \to \mathbb{C}$ is said to satisfy \emph{quasi-periodic boundary conditions} if
				\[
					\psi_{e_+}(v) = e^{-i \alpha} \psi_{e_-}(v), \qquad
					\psi_{e_+}'(v) = e^{-i \alpha} \psi_{e_-}'(v),
				\]
				where $\alpha$ is some real number. Note that $\alpha = 0$ reduces to the periodic boundary condition, which in fact corresponds with Kirchhoff-Neumann vertex conditions introduced in the previous section.
			\end{loopexample}

			In order to figure out how to generalize quasi-periodic boundary conditions to a general graph, let us consider now a quantum graph $G$ with $H = -\Delta$ and local vertex conditions of the form
			\begin{equation}\label{eq:quasiperiodic-BC1}
				\veval{\psi}_e(v) = e^{-i\veval{\chi}_e(v)} \veval{\psi}_{e_0}(v), \qquad (\forall e\in E_v),
			\end{equation}
			with $e_0 \in E_v$ an arbitrary reference edge, for every vertex $v \in V$. Note that since $e_0$ is the reference edge, it is implicit in \autoref{eq:quasiperiodic-BC1} that $\veval{\chi}_{e_0}(v) = 0$. As we stated in \autoref{subsec:quantum-wires}, vertex conditions involving only the vertex evaluation of the function (as opposed to those involving vertex evaluation of its derivative) fix the eigenspace of $U_v$ associated to -1: the projector over that eigenspace, $P_v$, must be such that the conditions on the function $\psi$ can be expressed as $P_v \veval{\psi}(v) = 0$. Proceeding as discussed in \autoref{subsec:quantum-wires} one concludes that $P_v^\perp = I - P_v$ is the rank-one projector over the span of $\veval{e^{-i\chi}}(v) \coloneqq \{e^{-i\veval{\chi}_e(v)}\}_{e \in E_v}$; that is,
			\[
				P_v = I - P_v^\perp, \qquad
				P_v^\perp = \frac{\veval{e^{-i\chi}}(v) \veval{e^{-i\chi}}(v)^\dagger}{\norm{\veval{e^{-i\chi}}(v)}_\V^2}.
			\]
			As a consequence of $P_v^\perp$ being rank-one, we have that all self-adjoint extensions of $-\Delta$ with the form we are interested in can be parametrized by
			\[ U_v = e^{i\delta} P_v^\perp - (I - P_v^\perp) = (e^{i\delta} + 1)P_v^\perp - I, \]
			where $e^{i\delta}$ is the eigenvalue of $U_v$ associated to $P_v^\perp$, with $\delta \in \R$. Those vertex conditions are what we call \emph{quasi-$\delta$-type vertex conditions}, because they generalize $\delta$-type vertex conditions in the same sense that quasi-periodic boundary conditions generalize periodic ones. That is, by allowing a phase shift of the functions on the vertices instead of enforcing continuity. Note that in the case $\delta = 0$ we would have \emph{quasi-Kirchhoff-Neumann vertex conditions} i.e., a generalization of Kirchhoff-Neumann vertex conditions that allows the functions to have a phase shift.

			As said before, the relationship between magnetic and standard Laplacians can be used to study the vertex conditions we have just introduced. Note that \autoref{eq:quasiperiodic-BC1} can be rewritten as
			\[
				\veval{\varphi}(v) = \veval{T}(v) \veval{\psi}(v),
			\]
			where again $\veval{T}(v)$ is the diagonal matrix $\veval{T} = \diag(\{e^{i \veval{\chi}_e(v)}\}_{e\in E_v})$ and the uniform vector $\veval{\varphi}(v) = \{\veval{\psi}_{e_0}(v)\}_{e \in E_v}$. This motivates the definition of a global transformation $T$ as in \autoref{corol:magnetic-standard-equiv} such that its restriction to each vertex is the operator $\veval{T}(v)$ defined above. Therefore, with the change $\psi \to \varphi = T \psi$ we can work with the magnetic Laplacian and the local vertex conditions given at each vertex by the unitary
			\[ \tilde{U}_v = \veval{T}(v) U_v \veval{T}(v)^{-1}. \]
			A simple computation leads to $\tilde{P}_v^\perp = \veval{T}(v) P_v^\perp \veval{T}(v)^{-1} = (\deg v)^{-1} \ones$, where $\ones$ is the square matrix of size $\deg v$ made of ones. That is, vertex conditions on $\varphi$ are the $\delta$-type vertex conditions\footnote{Note that now, since we are working with a magnetic Laplacian vertex conditions involve $\oveval{D \varphi}$ instead of $\veval{\dot{\varphi}}$.} described in \autoref{eq:delta-vertex-cond}.

			We have yet to define the global transformation $T$ we used above (see Eq. \eqref{eq:def-T}). At the moment, we have only fixed its restriction to the vertices and therefore we are left with some freedom to choose the functional form of $\chi = \{\chi_e\}_{e \in E}$. However, as stated by \autoref{corol:magnetic-Laplacian-constant}, the functional dependence of $A$ (and therefore that of $\chi$) has no relevance as the resulting magnetic Laplacian will be equivalent to one with edge-constant vector potential, both with the same vertex conditions. Thus, we chose for simplicity a linear dependence for $\chi$,
			\begin{equation}\label{eq:form-of-chi}
				\chi_e(x_e) = \veval{\chi}_e(\partial_- e) + \frac{1}{\l_e} \left[\veval{\chi}_e(\partial_+ e) - \veval{\chi}_e(\partial_- e)\right]x_e,
			\end{equation}
			and consequently an edge-constant magnetic vector potential $A = \{\l_e^{-1} [\chi_e(\partial_+ e) - \chi_e(\partial_-e)]\}_{e \in E}$. If we now define the edge-constant function $b = \{\veval{\chi}_e(\partial_-e)\}_{e \in E}$, we can denote shortly $\chi = Ax + b$.

			Let us apply these results to the Simple Loop Example.
			\begin{loopexample}\label{example:0-2}
				For simplicity we consider $\l_e = 1$ in the continuation of this example. In this case, we have obviously $\chi_{e_-}(v) = 0$, $\chi_{e_+}(v) = \alpha$ and $\delta = 0$. A simple computation shows that
				\[
					U_v = \begin{bmatrix}
						0			&	e^{-i\alpha}	\\
						e^{i\alpha}	&	0
					\end{bmatrix},
					\qquad
					\veval{T}(v) = \begin{bmatrix}
						1	&	0			\\
						0	&	e^{i\alpha}
					\end{bmatrix}.
				\]
				Finally, from Equation \eqref{eq:form-of-chi} it follows
				\[
					\chi_e(x_e) = \alpha x_e.
				\]

			\end{loopexample}

	\section{Time evolution}\label{sec:time-evolution}
		Now we have finished developing the tools we will be using to describe \emph{spatially} the systems we are working on. In this section we are going to focus on the system's evolution. Quantum systems' evolution is given by a Hamiltonian operator $H(t)$, which in the most general setting depends itself on the time $t$, and according to Schrödinger equation
		\begin{equation}\label{eq:schrodinger}
			i \Dt \psi(t) = H(t) \psi(t).
		\end{equation}
		In the case we are interested on, where the spatial description of the system is given by a metric graph $G$, $\{H(t)\}_t$ is a family of Hamiltonians in $L_2(G)$ and $\psi(t)$ is a curve in the state space, $\H = L_2(G)$.

		Concerning the existence of solutions for the Schrödinger equation with a given Hamiltonian, there are several results establishing conditions for solutions to exist \cite{kisynski_sur_1964, reed_methods_1975}. It is customary to search for solutions using the idea of unitary propagators, which are families of operators which allow us to write the solution of the Schrödinger equation with initial state $\psi_s$ at $t=s$ as $\psi(t) = U(t,s)\psi_s$ for $t > s$. A proper definition of a unitary propagator would be as follows:

		\begin{defn}\label{def:unitary-propagator}
			A two-parameter family of unitary operators $U(s,t)$, with $s,t \in \R$, that satisfies:
			\begin{enumerate}[label=\textit{(\roman*)},nosep]
				\item $U(r, s)U(s, t) = U(r, t)$
				\item $U(t,t) = I$
				\item $U(s, t)$ is jointly strongly continuous in $s$ and $t$
			\end{enumerate}
			is called a unitary propagator.
		\end{defn}

		After unitary propagators are introduced, the existence of the solutions for the associated Cauchy problems is equivalent to the existence of a unitary propagator for the \autoref{eq:schrodinger}. For the most general setting, in which $\dom H(t)$ varies with $t$, J. Kisyński gave conditions that $H(t)$ must satisfy for the unitary propagator to exist \cite{kisynski_sur_1964}. However, we will be interested in the less general case in which $\DD = \dom H(t)$ is the same for every $t$ and thus it is enough to consider a less general result by M. Reed and B. Simon \cite[\S X.12]{reed_methods_1975}. Instead of treating the case of families of self-adjoint operators, they study the more general case of families of generators of contraction semigroups, which can be directly applied to the case of families of self-adjoint families since for $H$ self-adjoint, $\pm iH$ is the generator of a contraction semigroup (see Theorem X.47a and Example 1 on \S{X}.8 of \cite{reed_methods_1975}). Let $A(t)$ denote a family of generators of a contraction semigroup. For such a case, Reed and Simon define an approximation for the propagator $U(t,s)$ solving the equation
		\[
			\Dt \varphi(t) = -A(t) \varphi(t), \qquad \varphi(s) = \varphi_s
		\]
		in the following way: considering a partition of the time interval, taking a generator which is constant on each element of the partition and providing conditions ensuring that it converges to solution. If, for example, the time interval we are interested in is $I = [0, 1]$, they take the partition made of $k$ elements $I_j = [\frac{j-1}{k}, \frac{j}{k}]$, $1 \leq j \leq k$, and define the approximate propagator
		\[
			U_k(t, s) = \exp\left(-i(t-s) A\left( \frac{j - 1}{k} \right)\right) \quad
			\text{if } \frac{j - 1}{k} \leq s \leq t \leq \frac{j}{k}
		\]
		and
		\[ U_k(t, s) = U_k(t, r)U_k(r, s)	\quad	\text{if } 0 \leq s \leq r \leq t \leq 1. \]
		That is, if $s,t$ lie in the same interval $I_j$ they consider the evolution operator given by the action of the contraction semigroup generated by $A(\frac{j-1}{k})$ and if $t,s$ lie in different intervals, they use the product property of the semigroup to define it.

		For convenience of the notation, assuming $0 \in \rho(A(t))$ for all $t$, they define a two-parameter family of operators
		\[ C(t, s) = A(t)A(s)^{-1} - I, \]
		which is a bounded operator by the Closed Graph Theorem. Moreover, note that $0 \in \rho(A(t))$ for all $t$ implies that $A(t)$ is a bijection of $\DD$ onto $\H$, and therefore for every $\varphi \in \H$ there exists $\psi \in \DD$ such that $\varphi = A(s) \psi$. Thus, for that $\varphi$, $C(t,s) \varphi = A(t) \psi - A(s)\psi$. That is, studying the behavior of $C(t,s)\varphi$ for any $\varphi \in \H$ can be understood as studying that of $A(t) \psi - A(s) \psi$ for any $\psi \in \DD$.
		Now we have the ingredients to state the result by M. Reed and B. Simon we have been advancing:\footnote{Actually, Reed and Simon state the theorem for general Banach spaces.}

		\begin{thm}[M. Reed and B. Simon, {\cite[Thm. X.70]{reed_methods_1975}}]\label{thm:reed-simon}
			Let $\H$ be a Hilbert space and let $I$ be an open interval in $\R$. For each $t \in I$, let $A(t)$ be the generator of a contraction semigroup on $\H$ so that $0 \in \rho(A(t))$ and
			\begin{enumerate}[label={\it (\alph*)},nosep]
				\item\label{enum:reedSimon-i} The $A(t)$ have common domain $\DD$.
				\item\label{enum:reedSimon-ii} For each $\varphi \in \H, (t - s)^{-1} C(t, s)\varphi$ is uniformly strongly continuous and uniformly bounded in $s$ and $t$ for $t \neq s$ lying in any fixed compact subinterval of $I$.
				\item\label{enum:reedSimon-iii} For each $\varphi \in \H$, $C(t)\varphi = \lim_{s \nearrow t} (t - s)^{-1}C(t, s)\varphi$ exists uniformly for $t$ in each compact subinterval and $C(t)$ is bounded and strongly continuous in $t$.
			\end{enumerate}
		Then for all $s \leq t$ in any compact subinterval of $I$ and any $\varphi \in \H$,
		\[ U(t, s) \varphi = \lim_{k \to \infty} U_k(t, s)\varphi \]
		exists uniformly in $s$ and $t$. Further, if $\varphi_s \in \DD$, then $\varphi(t) = U(t, s)\varphi_s$ is in $\DD$ for all $t$ and satisfies
		\[ \Dt \varphi(t) = -A(t) \varphi(t), \qquad \varphi(s) = \varphi_s \]
		and $\norm{\varphi(t)} \leq \norm{\varphi_s}$ for all $t \geq s$.
		\end{thm}


		In the following sections, we are going to work with Hamiltonians whose time-dependent structure can be written as
		\[ H(t) = \sum_{i = 1}^n f_i(t) H_i, \]
		with $f_i: \R \to \R$ containing all the time dependence and $H_i$ being constant self-adjoint operators. Applying \autoref{thm:reed-simon} to this type of Hamiltonians is the purpose of \autoref{thm:timedependent-linearCombination-Hamiltonian}, which establishes sufficient conditions to be fulfilled so that the existence of a unitary propagator is guaranteed.

		\begin{thm}\label{thm:timedependent-linearCombination-Hamiltonian}
			Let $\{H_i\}_{i=1}^n$ be a family of symmetric operators on a dense subset $\mathcal{D}$ of a Hilbert space $\H$ and let $f_i: I \subset \R \to \R$ be real valued functions for $1 \leq i \leq n$. Define the time-dependent operator
			\[
				H(t) = \sum_{i = 1}^n f_i(t) H_i, \qquad
				\dom H(t) = \DD.
			\]
			If it holds
			\begin{enumerate}[label=\textit{(\roman*)},nosep]
				\item\label{enum:H1-thm-linCombHamiltonian} $H(t)$ is self-adjoint,
				\item\label{enum:H2-thm-linCombHamiltonian} $f_i \in C^1(I)$ for every $i$, and
				\item\label{enum:H3-thm-linCombHamiltonian} there exists a $K > 0$ (not depending on $t$) such that for every $\psi \in \mathcal{D}$, $\norm{H_i \psi} \leq K (\norm{H(t)\psi} + \norm{\psi})$ for every $t \in I$,
			\end{enumerate}
			then there exists a strongly differentiable unitary propagator $U(t,s)$ with $s,t \in I$ such that, for any $\psi_s \in \DD$, $\psi(t) = U(t,s)\psi_s$ satisfies
			\[ \Dt \psi(t) = -i H(t) \psi(t), \qquad \psi(s) = \psi_s. \]
		\end{thm}

		Before we proof \autoref{thm:timedependent-linearCombination-Hamiltonian} it is useful to introduce the following lemmas.

		\begin{lemma}\label{lemma:1}
			Let $H(t)$ be as in \autoref{thm:timedependent-linearCombination-Hamiltonian} and define $A_i = iH_i$ and $\tilde{A}(t) = iH(t) + I$. Then, for every $\varphi \in \mathcal{H}$, it holds
			\[
				\norm{A_i \tilde{A}(t)^{-1} \varphi} \leq 3K \norm{\varphi},
			\]
			with $K$ independent of $\varphi$ and $t$.
		\end{lemma}
		\begin{proof}
			Since $H(t)$ is self-adjoint, the spectrum of $\tilde{A}(t)$ is a subset of $i \mathbb{R} + 1 = \{i \alpha + 1: \alpha \in \mathbb{R}\} \subset \mathbb{C}$. Thus $\tilde{A}^{-1}$ is bounded and maps $\mathcal{H}$ into $\mathcal{D}$.

			That said, this lemma is a direct consequence of hypothesis \autoref{enum:H3-thm-linCombHamiltonian} of \autoref{thm:timedependent-linearCombination-Hamiltonian}. For every $\psi \in \mathcal{D}$,
			\[
				\norm{A_i \psi} = \norm{H_i \psi} \leq K (\norm{H(t) \psi} + \norm{\psi})
				= K\norm{\tilde{A}(t)\psi - \psi} + K \norm{\psi}
				\leq K \left(\norm{\tilde{A}(t)\psi} + 2\norm{\psi} \right)
			\]
			for every $t$, and since $\psi = \tilde{A}(t)^{-1}\varphi \in \mathcal{D}$ for every $\varphi \in \mathcal{H}$ it holds:
			\[
				\norm{A_i \tilde{A}(t)^{-1}\varphi}
				\leq K \left(\norm{\varphi} + 2\norm{\tilde{A}(t)^{-1}\varphi} \right).
			\]
			The distance from 0 to $\sigma(\tilde{A}(t))$ is at least 1, and thus $\norm{\tilde{A}(t)^{-1}} \leq 1$. Hence,
			\[
				\norm{A_i \tilde{A}(t)^{-1}\varphi} \leq 3K \norm{\varphi}.
			\]
		\end{proof}
		\begin{lemma}\label{lemma:2}
			Let $H(t)$ be as in \autoref{thm:timedependent-linearCombination-Hamiltonian} and define $A_i = iH_i$ and $\tilde{A}(t) = iH(t) + I$. Then, for $s,t$ lying in a compact subset of $I$ and for every $\varphi \in \mathcal{H}$, $\lim_{t \to s} \tilde{A}(t)\tilde{A}(s)^{-1} \varphi = \varphi$ uniformly on $s$.
		\end{lemma}
		\begin{proof}
			We need to prove the limit uniformly on $s$, i.e., that
			\[
				\lim_{t \to s} \norm{\tilde{A}(t)\tilde{A}(s)^{-1} \varphi - \varphi} = 0
			\]
			uniformly on $s$. We have
			\[
				\norm{\tilde{A}(t)\tilde{A}(s)^{-1} \varphi - \varphi}
				= \norm{\sum_{i=1}^n[f_i(t) - f_i(s)]A_i \tilde{A}(s)^{-1} \varphi}
				\leq \sum_{i=1}^n |f_i(t) - f_i(s)| \norm{A_i \tilde{A}(s)^{-1}\varphi}.
			\]
			Now, using \autoref{lemma:1} one gets
			\[
				\norm{\tilde{A}(t)\tilde{A}(s)^{-1} \varphi - \varphi}
				\leq 3K \sum_{i=1}^n |f_i(t) - f_i(s)| \norm{\varphi}.
			\]
			Hence, $\lim_{t \to s} \norm{\tilde{A}(t)\tilde{A}(s)^{-1} \varphi - \varphi} = 0$ uniformly on $s$ since every $f_i$ is uniformly continuous on $I$ (which follows from \autoref{enum:H2-thm-linCombHamiltonian} and the fact that we are considering a fixed, compact subset of $I$).
		\end{proof}

		\begin{proof}[Proof of \autoref{thm:timedependent-linearCombination-Hamiltonian}]
			This theorem is a consequence of \autoref{thm:reed-simon} and the fact that $A(t) = i H(t)$ is the generator of a contraction semigroup by Hille--Yoshida theorem (see \cite{reed_methods_1975}, Theorem X.47a and Example 1 on \S{X}.8).
			In order to apply \autoref{thm:reed-simon} we need to have $0 \in \rho(iH(t))$ for every $t$, which is not satisfied in general. However, since $H(t)$ is self-adjoint $i \in \rho(H(t))$ for every $t$ and therefore $-1 \in \rho(A(t))$ which implies $\tilde{A} = A(t) + I$ has 0 in its resolvent set. Note that if $\varphi(t) = \tilde{U}(t,s) \xi$ satisfies
			\[ \Dt \varphi(t) = -\tilde{A}(t) \varphi(t), \qquad \varphi(s) = \xi,\]
			then $\psi(t) = U(t,s) \xi$ with $U(t,s) \coloneqq \tilde{U}(t,s) e^{-i(s-t)}$ satisfies, by the product rule,
			\[ \Dt \psi(t) = - A(t) \psi(t), \qquad \psi(s) = \xi. \]
			Thus, existence of $\tilde{U}(t,s)$ with the properties in the statement of \autoref{thm:timedependent-linearCombination-Hamiltonian} guarantee the existence of $U(t,s)$ with the same properties.

			Hence, it is enough to show that $\tilde{A}(t)$ satisfies the hypothesis of \autoref{thm:reed-simon}. It is clear that $\tilde{A}(t)$ can be written as
			\[
				\tilde{A}(t) = I + i\sum_{j=1}^{n} f_j(t) H_j
				= \sum_{j=1}^{n+1} f_j(t) A_j
			\]
			with $A_i = iH_j$, $f_{n+1} = 1$ and $A_{n+1} = I$. Also is easy to check using Hille--Yoshida theorem that $\tilde{A}(t)$ is the generator of a contraction semigroup.

			Hypothesis \autoref{enum:reedSimon-i} of \autoref{thm:reed-simon} is satisfied since, by definition, every $H(t)$ (and therefore $\tilde{A}(t)$) has the same domain $\DD$.

			Regarding \autoref{enum:reedSimon-ii} and \autoref{enum:reedSimon-iii}, it is useful to write
			\begin{equation}\label{eq:thm-5.4-1}
				(t-s)^{-1}C(t,s) = (t-s)^{-1} [\tilde{A}(t) - \tilde{A}(s)] \tilde{A}(s)^{-1}
				= \sum_{i=1}^{n} \frac{f_i(t) - f_i(s)}{t - s} A_i \tilde{A}(s)^{-1}.
			\end{equation}
			For the shake of notation, let us denote $g_i(t,s) = \frac{f_i(t) - f_i(s)}{t-s}$, which clearly is $C^1$ in $s$ and $t$ for $t \neq s$ in $I$. Moreover, for $s\neq t$ lying in any fixed compact subinterval of $I$, $g_i$ is uniformly continuous because $f_i(t)$ is $C^1(I)$.

			From the previous equation it follows that
			\[
				\norm{(t-s)^{-1}C(t,s)\varphi}
				\leq \sum_{i=1}^{n+1} |g_i(t,s)| \norm{A_i \tilde{A}(s)^{-1}\varphi}
				\leq 3K \sum_{i=1}^{n+1} |g_i(t,s)| \norm{\varphi},
			\]
			where we have used \autoref{lemma:1} in the last inequality. For $s\neq t$ lying in any fixed compact subinterval of $I$, $|g_i(t,s)|$ is bounded uniformly on $s$ and $t$ since it is continuous and thus $\norm{(t-s)^{-1}C(t,s)\varphi}$ is uniformly bounded for such $s,t$.

			For the uniform strong continuity respect to $t$, it is clear that
			\[\begin{alignedat}{2}
				\norm{(t_0-s)^{-1}C(t_0,s)\varphi - (t-s)^{-1}C(t,s)\varphi}
				& \leq \sum_{i=1}^{n+1} |g_i(t_0,s) - g_i(t,s)| \norm{A_i \tilde{A}(s)^{-1}\varphi} \\
				& \leq 3K \sum_{i=1}^{n+1} |g_i(t_0,s) - g_i(t,s)| \norm{\varphi}
			\end{alignedat}\]
			and, thus, uniform continuity of $t \mapsto g_i(t,s)$ implies uniform strong continuity of the operator-valued function $t \mapsto (t-s)^{-1}C(t,s)$.

			On the other hand, regarding uniform strong continuity respect to $s$ we have
			\begin{equation}\label{eq:thm-5.4-2}
			\begin{alignedat}{2}
				\norm{(t-s_0)^{-1}C(t,s_0)\varphi - (t-s)^{-1}C(t,s)\varphi}
				&\leq \sum_{i=1}^{n+1} \norm{g_i(t,s_0) A_i\tilde{A}(s_0)^{-1}\varphi - g_i(t,s) A_i \tilde{A}(s)^{-1} \varphi}  \\
				&\leq \sum_{i=1}^{n+1} |g_i(t,s_0)| \norm{A_i\tilde{A}(s_0)^{-1}\varphi - A_i \tilde{A}(s)^{-1} \varphi} + \\
				& \phantom{leq} + \sum_{i=1}^{n+1}|g_i(t,s_0) - g_i(t,s)| \norm{A_i\tilde{A}(s)^{-1}\varphi}.
			\end{alignedat}
			\end{equation}
			Let us examine separately the two terms on the right-hand side. First,
			\[
				\sum_{i=1}^{n+1}|g_i(t,s_0) - g_i(t,s)| \norm{A_i \tilde{A}(s)^{-1}\varphi}
				\leq 3K \sum_{i=1}^{n+1}|g_i(t,s_0) - g_i(t,s)| \norm{\varphi}
			\]
			and therefore because $g_i$ is uniformly continuous for $s\neq t$ in a compact subinterval of $I$, for every $s \neq t$ and every $\varepsilon > 0$ there exists $\delta_1 > 0$ such that for $|s_0 - s| < \delta_1$ it holds
			\[
				\sum_{i=1}^{n+1}|g_i(t,s_0) - g_i(t,s)| \norm{A_i \tilde{A}(s)^{-1}\varphi}
				\leq \frac{\varepsilon}{2}.
			\]

			For the first term in Eq. \eqref{eq:thm-5.4-2},
			\[
				\sum_{i=1}^{n+1} |g_i(t,s_0)| \norm{A_i \tilde{A}(s_0)^{-1}\varphi - A_i \tilde{A}(s)^{-1} \varphi} \leq 3K\sum_{i=1}^{n+1} |g_i(t,s_0)| \norm{\tilde{A}(s)\tilde{A}(s_0)^{-1}\varphi - \varphi}
			\]
			and thus by \autoref{lemma:2} and the fact that $g_i$ is uniformly bounded for $s \neq t$, for every $s \neq t$ and every $\varepsilon > 0$ there exists $\delta_2 > 0$ such that for $|s_0 - s| < \delta_2$ it holds
			\[
				\sum_{i=1}^{n+1} |g_i(t,s_0)| \norm{A_i \tilde{A}(s_0)^{-1}\varphi - A_i \tilde{A}(s)^{-1} \varphi} \leq \frac{\varepsilon}{2}.
			\]
			Hence, taking $\delta = \min\{\delta_1, \delta_2\}$ and substituting into Eq. \eqref{eq:thm-5.4-2} we have that for $|s_0 - s| < \delta$
			\[
				\norm{(t-s_0)^{-1}C(t,s_0)\varphi - (t-s)^{-1}C(t,s)\varphi} \leq \varepsilon
			\]
			which shows that hypothesis \autoref{enum:reedSimon-ii} is fulfilled.

			Regarding hypothesis \autoref{enum:reedSimon-iii} of \autoref{thm:reed-simon}, it is easy to see that $C(t) \varphi = \sum_{i=1}^{n+1} f_i'(t) A_i \tilde{A}(t)^{-1} \varphi$. Indeed, from \autoref{eq:thm-5.4-1} we get
			\[\begin{alignedat}{2}
				\norm{(t - s)^{-1} C(t,s) \varphi - \sum_{i=1}^{n+1} f_i'(t) A_i \tilde{A}(t)^{-1} \varphi}
				&= \norm{\sum_{i=1}^{n+1} \left[g_i(t,s) A_i \tilde{A}(s)^{-1} - f_i'(t) A_i \tilde{A}(t)^{-1}\right]\varphi} \\
				& \leq \sum_{i=1}^{n+1} |f_i'(t)|\norm{A_i\tilde{A}(s)^{-1}\varphi - A_i\tilde{A}(t)^{-1}\varphi} + \\
				&\phantom{\leq} + \sum_{i=1}^{n+1} |g_i(t,s) - f_i'(t)| \norm{A_i \tilde{A}(s)^{-1}\varphi}.
			\end{alignedat}\]
			Using again \autoref{lemma:1}, \autoref{lemma:2} and the fact that we are considering a compact subinterval, the continuity of every $f_i'$ and the definition of derivative implies the limit $C(t)\varphi = \lim_{s \to t} (t-s)^{-1} C(t,s)\varphi$ exists uniformly on $t$.

			Boundedness of $C(t)$ as an operator follows directly from \autoref{lemma:1} and continuity of $f_i'(t)$:
			\[
				\norm{C(t) \varphi} = \norm{\sum_{i=1}^{n+1} f_i'(t) A_i \tilde{A}(t)^{-1} \varphi}
				\leq \sum_{i=1}^{n+1} |f_i'(t)| \norm{A_i \tilde{A}(t)^{-1} \varphi}
				\leq 3K \sum_{i=1}^{n+1} |f_i'(t)| \norm{\varphi}.
			\]
		\end{proof}

		Besides existence of unitary propagators for Schrödinger equations associated with Hamiltonians of the type we are dealing with, we are going to need a result on how close the evolution induced by two of these Hamiltonians is when they are similar (in the precise sense introduced in \autoref{thm:aprox-Hamiltonians-aprox-sol}).

		\begin{thm}\label{thm:aprox-Hamiltonians-aprox-sol}
			Let $H_1(t) = \sum_{i = 1}^n f_i(t) H_i$ and $H_2(t) = \sum_{i = 1}^n g_i(t) H_i$, with common domain $\DD$, both satisfying the hypothesis of \autoref{thm:timedependent-linearCombination-Hamiltonian}. Then, for every $\psi \in \DD$, every $T > 0$ and every $\varepsilon > 0$ there exist $\delta_1, \delta_2, \dots, \delta_n > 0$ such that $\norm{f_i - g_i}_\infty < \delta_i$ implies $\norm{U_1(T, s) \psi - U_2(T, s) \psi} < \varepsilon$.
		\end{thm}
		\begin{proof}
			By \autoref{thm:timedependent-linearCombination-Hamiltonian}, there exist unitary propagators $U_1(t, s)$, $U_2(t, s)$ associated with $H_1(t)$, $H_2(t)$ respectively. Since for any $\psi \in \DD$, $t \mapsto U_\l(t, s) \psi$ is strongly differentiable ($\l = 1, 2$), we have $t \mapsto \norm{U_1 (t, s) \psi - U_2 (t, s) \psi}$ is differentiable and by the Fundamental Theorem of Calculus we have
			\[
				\norm{U_1 (T, s) \psi - U_2 (T, s) \psi}
				= \int_s^T \Dt \norm{U_1 (t, s) \psi - U_2 (t, s) \psi} dt.
			\]
			Strong differentiability implies that we can take the derivative into the norm and get
			\[ \begin{alignedat}{2}
				\norm{U_1 (T, s) \psi - U_2 (T, s) \psi}
				&= \int_s^T \norm{\Dt U_1 (t, s) \psi - \Dt U_2 (t, s) \psi} dt \\
				&= \int_s^T \norm{H_1(t) \psi - H_2(t) \psi} dt \\
				&\leq \sum_{i=1}^n \int_s^T \abs{f_i - g_i} \norm{H_i \psi} dt \\
				& \leq (T - s) \sum_{i=1}^n \norm{f_i - g_i}_\infty \norm{H_i \psi}
			\end{alignedat} \]
			Hence, it is enough to take
			\[ \delta_i = \frac{\varepsilon}{n (T - s) \norm{H_i \psi}}. \]
		\end{proof}

	\section{Boundary driven dynamics: a particular case} \label{sec:boundary-driven-dynamics}
		\subsection{Quasi-$\delta$-type boundary control systems}
			As pointed out in the introduction, the main goal of this work is to explore the possibility of controlling a free quantum particle by \emph{boundary controls}, that is, taking the space of self-adjoint extensions of the Laplacian as the space of controls. In this section we are going to prove that, for some cases, considering only a subset of the self-adjoint extensions space is enough to control the system. In particular, the one associated with time dependent quasi-$\delta$-type vertex conditions.

			Thus, we consider a family of quantum graphs whose Hamiltonians are standard Laplacians with time-dependent quasi-$\delta$-type vertex conditions such that
			\begin{equation}\label{eq:quasiperiodic-timedependent}
				\veval{\psi}_e(v) = e^{-i\veval{\chi}_e(v, t)} \veval{\psi}_{e_0}(v), \qquad (\forall e\in E_v),
			\end{equation}
			where again $e_0$ is an arbitrary reference edge in $E_v$ but now $\chi(t) = \{\veval{\chi}_e(v, t)\}_{e\in E,v\in V}$ forms a family of functions from $G$ to $\R$. One should notice that the time dependence of this Hamiltonians is subtle: usually one faces the problem where $\dom H(t)$ does not depend on time but the explicit, functional form of $H(t)$ does, while here we have $-\Delta$ for every $t$ and $\dom \Delta$ varying with time. That is, we are considering at each time a different self-adjoint extension of the Laplacian on our quantum graph. Therefore, looking for solutions of Schrödinger equation is harder than usual, but based on the equivalence we established in \autoref{subsec:quasi-periodic-BC} we will be able to transform some of these problems into equivalent ones with Hamiltonian $H(t)$ such that $\dom H(t)$ remains constant and time dependence appears explicitly in the form of $H(t)$.

			Now, if we are able to choose the way the time-dependent quasi-$\delta$-type vertex conditions change (i.e., if we can freely choose the functions $\chi(t)$) we can try to induce a state transition on the system choosing an appropriate time evolution for the boundary conditions. The control problem we are going to study is the following.
			\begin{defn}\label{def:quasi-periodic-BCS}
				Given a metric graph $G$, we call quasi-$\delta$-type boundary control system associated to $G$ to the family of quantum graphs $\{(G, -\Delta(t))\}_t$, with control $t \mapsto \Delta(t)$ a smooth curve on the space of self-adjoint extensions of $\Delta$ with quasi-$\delta$-type vertex conditions.
			\end{defn}

			Following the ideas exposed in Subsection \ref{subsec:quasi-periodic-BC}, we can find the natural equivalence between a quasi-$\delta$-type boundary control system and a regular magnetic controlled one:
			\begin{prop}\label{prop:quasi-periodic-BC-system-equivalence}
				Every quasi-$\delta$-type boundary control system is (unitarily) equivalent to a magnetic control system, that is, a system whose evolution is given by the Hamiltonian
				\[ H(t) = -\left[\left(\Dx - i A(t)\right)^2 + b'(t) + A'(t) x \right] \]
				with $\delta$-type vertex conditions and controls $A, b : I \subset \R \to \tilde{H}^1(G)$.
			\end{prop}
			\begin{proof}
				Take the family of unitary transformations $T(t)$ as in \eqref{eq:def-T} with $\chi = \chi(t)$:
				\[
					T(t): \{\psi_e\}_{e \in E} \in L_2(G) \mapsto \{e^{i \chi_e(t)} \psi_e\}_{e \in E} \in L_2(G).
				\]
				and define $\varphi(t) = T(t)\psi(t)$. The chain rule implies\footnote{As pointed out in \autoref{subsec:preliminaries-hilbert}, here and everywhere in this work whenever a derivative of an operator-valued function appears, it must be understood in the strong sense.}
				\[ \Dt \varphi(t) = \D{T}{t}(t) \psi(t) + T(t) \Dt \psi(t) \]
				and therefore using the Schrödinger equation \eqref{eq:schrodinger} for $\psi$ and the definition of $T(t)$, we have
				\[ i \Dt \varphi(t) = \left[- \Dt \chi(t) - T(t) \Delta T(t)^{-1} \right] \varphi(t). \]
				If we take $\chi$ as in Equation \ref{eq:form-of-chi}, we find that the time dependence on the boundary conditions appears as a time dependence of the edge-constant functions $A$, $b$ defined there. Thus, we have $\chi(t) = A(t) x + b(t)$ and $T(t)\Delta T(t)^{-1} = -\left(\Dx - iA(t)\right)^2 \eqqcolon -D^2$ with $\delta$-type vertex conditions (see \autoref{eq:delta-vertex-cond}). Substituting this into the previous equation, it follows
				\[
					i \Dt \varphi(t) = -\left[\left(\Dx - i A(t)\right)^2 + b'(t) + A'(t) x \right] \varphi(t),
				\]
				with local vertex conditions at $v \in V$:
				\[ \left\{ \begin{alignedat}{2}
					&\varphi(t) \text{ continuous in } v, \\
					&\sum_{e \in E_v} \oveval{D \varphi}_e(v) = - \deg v \tan \frac{\delta}{2} \varphi(v).
				\end{alignedat} \right. \]
			\end{proof}

			This proposition shows how to treat the quasi-$\delta$-type boundary control system applying a unitary transformation which leads to a magnetic controlled system, where the time-dependence of the Hamiltonian's domain has been reduced or even removed. Note that the part of the vertex conditions related to the \emph{derivative} carries an implicit time-dependent term since $D = \left(\Dx - i A(t)\right)$. Hence, the part of the vertex conditions related to the \emph{derivative} is
			\[
				0 = \sum_{e \in E_v} \oveval{D\varphi}_e
				= \sum_{e \in E_v} \veval{\dot{\varphi}}_e + \sum_{e \in E_v} \oveval{A}_e(v,t),
			\]
			which makes clear that the domain of the family $H(t)$ in \autoref{prop:quasi-periodic-BC-system-equivalence} is constant only if for every vertex $v \in V$, $\sum_{e \in E_v} \oveval{A}_e(v,t) = 0$ for every $t$. From now on, we only consider quasi-$\delta$-type boundary control systems that satisfy this condition, and some particular examples will be exposed at the end of this section. For further references, let us fix a name for such systems.
			\begin{defn}\label{defn:simple-quasi-periodic-control-system}
				We say that an edge-constant magnetic vector potential $A$ is simple for a given graph $G$ if for every vertex $v$ it satisfies
				\[
					\sum_{e \in E_v} \oveval{A}_e(v) = 0.
				\]
				We say that a quasi-$\delta$-type boundary control system is simple if its underlying graph admits a simple vector potential
			\end{defn}

			\begin{loopexample}\label{example:0-3}
				Continuing the simple loop example series, let us consider the graph $G_0$ showed in \autoref{fig:example0} and let the $A=\alpha$ in part \ref{example:0-1} be a function of time. Then $\chi(t,x) = \alpha(t)x / \l_e$ and \autoref{prop:quasi-periodic-BC-system-equivalence} gives the equivalent magnetic Hamiltonian
				\[
					H(t) = -\left[\left(\Dx - i\alpha(t)\right)^{2} + \alpha'(t) x\right].
				\]
				Note that this is a simple quasi-$\delta$-type boundary control system for any $\alpha$ because
				\[
					\sum_{e \in E_v} \oveval{A}_e(v) = A_{e_+}(v) - A_{e_-}(v) = \alpha - \alpha = 0.
				\]
			\end{loopexample}

		\subsection{The controllability problem}
			The controllability problem consists on answering whether \emph{any transition} on the system can be induced by choosing an appropriate curve $\chi(t)$ or not. As we pointed out in \autoref{subsec:controllability-notions}, we address approximate controllability and thus, by \emph{induce any transition}, what we actually mean is \emph{reach as close as desired to a target state from a given one}.

			Relying on \autoref{prop:quasi-periodic-BC-system-equivalence}, approximate controllability of the system by choosing the way quasi-$\delta$-type vertex conditions change is equivalent to approximate controllability of a quantum system with Hamiltonian
			\[ H(t) = -\left[\left(\Dx - i A(t)\right)^2 + b'(t) + A'(t) x \right] \]
			and $\delta$-type vertex conditions.

			Unlike control on a finite dimensional Hilbert space, control on an infinite dimensional Hilbert space has no general result giving necessary and sufficient conditions for (approximate) controllability. However, it will be enough for us to base on a theorem by T. Chambrion et al.\ \cite{chambrion_controllability_2009}, giving sufficient conditions to prove approximate controllability for the quasi-$\delta$-type boundary control system associated to some graphs. In the referenced work, Chambrion et al.\ study the approximate controllability of some bilinear control systems; that is, systems whose evolution is given by
			\begin{equation}\label{eq:bilinear-schrodinger-control}
				i\Dt \psi(t) = H_0 \psi(t) + u(t) H_1 \psi(t),
			\end{equation}
			with $u: \R \to (0, c)$. Moreover, they assume that:
			\begin{enumerate}[%
					label=\textit{(A\arabic*)},
					nosep,
					labelindent=0.5\parindent,
					leftmargin=*
					]
				\item\label{H1} $H_0, H_1$ are self-adjoint operators not depending on $t$,
				\item\label{H2} there exists an orthonormal basis $\{\phi_n\}_{n \in \N}$ of $\H$ made of eigenvectors of $H_0$, and
				\item\label{H3} $\phi_n \in \dom H_1$ for every $n \in \N$.
			\end{enumerate}
			Bilinear control systems satisfying conditions \autoref{H1}--\autoref{H3} will be called normal bilinear systems. For them, the following theorem is proven:
			\begin{thm}[Chambrion et al.\ {\cite[Thm. 2.4]{chambrion_controllability_2009}}.] \label{thm:chambrion-controllability}
				Consider a normal bilinear quantum control system, with $c > 0$ as described above. Let $\{\lambda_n\}_{n \in \N}$ denote the eigenvalues of $H_0$, each of them associated to the eigenfunction $\phi_n$. Then, if the elements of the sequence $\{\lambda_{n+1} - \lambda_n\}_{n \in \N}$ are $\Q$-linearly independent and if $\langle \phi_{n+1}, H_1 \phi_n \rangle \neq 0$ for every $n \in \N$, the system is approximately controllable.
			\end{thm}

			Based on this theorem we will prove the main result of this work which ensures the approximate controllability of the quasi-$\delta$-type boundary control problem. This is the first instance in which controllability of a system using boundary controls is considered. Before doing that it is convenient to introduce the following lemma:
			\begin{lemma}\label{lemma:chambrion-thm-applies}
				Consider a normal bilinear quantum control system $i \Dt \psi = H_0 \psi + u(t) H_1 \psi$ with $H_0$, $H_1$ such that $H(t) = H_0 + u(t)H_1$ satisfies hypothesis of \autoref{thm:timedependent-linearCombination-Hamiltonian}. Then given any $\varepsilon > 0$ there exist perturbed Hamiltonians $\tilde{H}_0$, $\tilde{H}_1$ with the same domain as $H_0$ such that they satisfy the conditions of \autoref{thm:timedependent-linearCombination-Hamiltonian} and also those of \autoref{thm:chambrion-controllability} and such that for every $t > s$ and every $C^1$ piecewise function $u:[s,t] \to \R$, it holds:
				\[
					\norm{U(t, s) \psi - \tilde{U}(t,s) \psi} < \varepsilon, \qquad
					(\forall \psi \in \dom H_0),
				\]
				where we denote by $U(t,s)$ and $\tilde{U}(t,s)$ the unitary propagators associated with $H(t)$ and $\tilde{H}(t) = \tilde{H}_0 + u(t) \tilde{H}_1$ respectively.
			\end{lemma}
			\begin{proof}
				Let $\lambda_k, \phi_k$ denote the eigenpairs of $H_0$. If it is such that $\Q$-linear independence condition of \autoref{thm:chambrion-controllability} hold, we set $\tilde{H}_0 = H_0$; otherwise take an increasing sequence of positive irrational numbers $\nu_k$ such that they are rationally independent and $\nu_k < 2^{-k}$. Then define
				\[ H_{0,p} = \sum_{k \in \N} \nu_k \phi_k \phi_k^\dagger. \]
				Obviously $H_{0,p}$ has the same domain as $H_0$ because $\norm{H_{0,p}\psi}^2 \leq \sum_k 2^{-2k}  |\langle \phi_k, \psi \rangle |^2$, from what we get $\norm{H_{0,p}\psi} \leq \norm{\psi}$ and $H_{0,p}$ can be chosen to have the same domain as $H_0$.

				Define $\tilde{H}_0 = H_0 + \mu_0 H_{0,p}$ with $\mu_0 \in \Q$. Then $\tilde{H}_0$ has eigenvalues $\lambda_k + \mu_0 \nu_k$ satisfying the rationally independence condition of \autoref{thm:chambrion-controllability}.

				If $H_1$ is such that $\langle \phi_{n+1}, H_1 \phi_n \rangle = 0$ for $n \in N \subset \N$, take a sequence of positive, non-vanishing terms $\{\alpha_n\}_{n \in N}$ such that $\alpha_n < 2^{-n}$ and define
				\[ H_{1,p} = \sum_{n \in N} \alpha_n \phi_{n+1} \phi_n^\dagger. \]
				Again the domain of $H_{1,p}$ can be chosen to be $\dom H_0$, since $\norm{H_{1,p} \psi}^2 \leq \sum_{n \in N} 2^{-2n} |\langle \phi_n, \psi \rangle|^2$ and thus $\norm{H_{1,p}\psi} \leq \norm{\psi}$.

				Defining $\tilde{H}_1 = H_1 + \mu_1 H_{1,p}$ with $\mu_1$ real it is clear that $\tilde{H}_1$ satisfies $\langle \phi_{n+1}, \tilde{H}_1 \phi_n \rangle \neq 0$ for any $n \in \N$.

				From what we already said, taking into account that $H_{0,p}$ and $H_{1,p}$ are bounded, it follows that if $H(t)$ satisfies the hypothesis of \autoref{thm:aprox-Hamiltonians-aprox-sol}, so does $\tilde{H}(t)$ on each interval in which $u(\tau)$ is $C^1$ and therefore, taking $\mu_0$ and $\mu_1$ small enough we have
				\[
					\norm{U(t, s) \psi - \tilde{U}(t,s) \psi} < \varepsilon, \qquad
					\text{for all } \psi \in \dom H_0.
				\]
			\end{proof}

			Before introducing the next theorem we will proof the following lemma, wich shows that the families of Hamiltonians $H(t)$ that we consider (i.e., those on \autoref{prop:quasi-periodic-BC-system-equivalence}) satisfy hypothesis $\autoref{enum:H3-thm-linCombHamiltonian}$ of \autoref{thm:timedependent-linearCombination-Hamiltonian} and therefore have well defined evolutions. It is clear that if $A$ is a simple magnetic vector potential, $H(t)$ fulfills all the hypothesis but \autoref{enum:H3-thm-linCombHamiltonian}, which requires some work to prove. The following result shows that hypothesis $\autoref{enum:H3-thm-linCombHamiltonian}$ also holds and thus \autoref{thm:timedependent-linearCombination-Hamiltonian} can be actually applied to $H(t)$.

			\begin{lemma}\label{lemma:magnetic-Laplacian-satisfy-thm5.3}
			    Let $A: G \to \mathbb{R}$ be an edge-constant magnetic vector potential over a metric graph $G$, and denote by $-D_A^2$ a self-adjoint extension of the associated magnetic Laplacian, whose domain we denote by $\mathcal{D} \subset \tilde{H}^2(G)$. For every $r > 0$, there exists a constant $K$ (not depending on $A$) such that if $\max_{e \in E} |A_e| < r$ then
			    \[
			        \norm{\D[2]{\psi}{x}} \leq K \left(\norm{D_A^2 \psi} + \norm{\psi}\right)
			    \]
			    for all $\psi \in \mathcal{D}$.
			\end{lemma}
			\begin{proof}
			    We will prove this lemma in three steps. First, we will show that for every vector potential $A: G \to \mathbb{R}$ the bound in the lemma stands with constant $K_A$ depending on $A$. In fact, since the magnetic Laplacian is self-adjoint, $-D_A^2 + iI$ is invertible with bounded inverse. Also $-\Dx[2]$ is closed in $\tilde{H}^2(G)$ since it is the adjoint of the standard Laplacian with the minimal symmetric domain ($i.e.$, with domain $\mathcal{D}_0 = \{\psi \in \tilde{H}^2(G): \veval{\psi} = 0, \veval{\dot{\psi}} = 0\}$) as is well-known (see, e.g., \cite{lions_problemes_1968}). This two facts imply $\Dx[2] (D_A^2 + iI)^{-1}$ is a bounded operator in $L_2(G)$. Therefore, for any $\psi \in \mathcal{D}$,
			    \begin{equation}\label{eq:magneticLaplacian-lowerBound-1}
			        \norm{\D[2]{\psi}{x}} = \norm{\Dx[2] (D_A^2 + iI)^{-1} (D_A^2 + iI) \psi}
			        \leq K_A \norm{(D_A^2 + iI) \psi} \leq K_A \left(\norm{D_A^2 \psi} + \norm{\psi}\right)
			    \end{equation}
			    where we have defined $K_A = \norm{\Dx[2] (D_A^2 + iI)^{-1}}$.

			    Once we have proved Equation \eqref{eq:magneticLaplacian-lowerBound-1}, we can prove that for every edge-constant vector potential $A$ there exists an $\varepsilon_A > 0$ such that for any edge-constant magnetic vector potential $B$ satisfying $\max_{e \in E} |A_e - B_e| \leq \varepsilon_A$ it holds
			    \[
			        \norm{\D[2]{\psi}{x}} \leq \tilde{K}_A \left(\norm{D_B^2 \psi}+\norm{\psi}\right),
			    \]
			    with $\tilde{K}_A > 0$ not depending on $B$. Indeed, from Equation \eqref{eq:magneticLaplacian-lowerBound-1} we have
			    \begin{equation}\label{eq:magneticLaplacian-lowerBound-2}
			        \norm{\D[2]{\psi}{x}} \leq K_A \left(\norm{D_A^2 \psi} + \norm{\psi}\right)
			        \leq K_A \left(\norm{D_A^2 \psi - D_B^2 \psi} + \norm{D_B^2\psi} + \norm{\psi} \right).
			    \end{equation}
			    Let us examine the first term in the parenthesis. By definition of the norm in $L_2(G)$,
			    \begin{equation}\label{eq:magneticLaplacian-lowerBound-3}
			        \begin{alignedat}{2}
			            \norm{D_A^2 \psi - D_B^2 \psi}^2
			            &= \sum_{e \in E} \norm{(D_A^2 \psi)_e - (D_B^2 \psi)_e}^2_{L_2(I_e)} \\
			            &= \sum_{e \in E} \norm{(A_e^2 - B_e^2) \psi_e + 2i(A_e - B_e)\D{\psi_e}{x_e}}^2_{L_2(I_e)}.
			        \end{alignedat}
			    \end{equation}
			    Denoting $\varepsilon_e = |A_e - B_e|$ and using the triangle inequality one gets
			    \[\begin{alignedat}{2}
			        \norm{(A_e^2 - B_e^2) \psi_e + 2i(A_e - B_e)\D{\psi_e}{x_e}}_{L_2(I_e)}
			        &\leq \varepsilon_e(2|A_e| + \varepsilon_e) \norm{\psi_e}_{L_2(I_e)} + 2 \varepsilon_e \norm{\D{\psi_e}{x_e}}_{L_2(I_e)} \\
			        &\leq \varepsilon(2\max_e|A_e| + \varepsilon) \norm{\psi_e}_{L_2(I_e)} + 2 \varepsilon \norm{\D{\psi_e}{x_e}}_{L_2(I_e)},
			    \end{alignedat}\]
			    where we take $\varepsilon = \max_e \varepsilon_e$.
			    Now, using the well-known fact \cite[Thm. 5.2]{adams_sobolev_2003} that
			    \[
			        \norm{\D{\psi_e}{x_e}}_{L_2(I_e)} \leq \tilde{K} \left(\norm{\D[2]{\psi_e}{x_e}}_{L_2(I_e)} + \norm{\psi_e}_{L_2(I_e)} \right),
			    \]
			    we get
			    \[\begin{alignedat}{2}
			        \norm{(A_e^2 - B_e^2) \psi_e + 2i(A_e - B_e)\D{\psi_e}{x_e}}_{L_2(I_e)}
			        &\leq \varepsilon(2\max_e|A_e| + \varepsilon + 2 \tilde{K}) \norm{\psi_e}_{L_2(I_e)} + 2 \varepsilon \tilde{K} \norm{\D[2]{\psi_e}{x_e}}_{L_2(I_e)} \\
			        &\leq \varepsilon(2\max_e|A_e| + \varepsilon + 2 \tilde{K}) \left(
			            \norm{\psi} + \norm{\D[2]{\psi}{x}}
			        \right),
			    \end{alignedat}\]
			    where we have used that $\norm{\psi}^2 = \sum_{e \in E} \norm{\psi_e}^2$.

			    Let us denote $\kappa(\varepsilon) = \varepsilon(2\max_e|A_e| + \varepsilon + 2 \tilde{K})$, which is a continuous monotone function of $\varepsilon$ with range $[0, \infty)$. Substituting into \autoref{eq:magneticLaplacian-lowerBound-3} and taking the square root
			    \[
			        \norm{D_A^2 \psi - D_B^2 \psi} \leq \kappa(\varepsilon)\sqrt{|E|} \left(
			            \norm{\psi} + \norm{\D[2]{\psi}{x}}
			        \right).
			    \]

			    Hence, from Equation \eqref{eq:magneticLaplacian-lowerBound-2} we get
			    \[
			        \left(1 - \kappa(\varepsilon) K_A\sqrt{|E|}\right) \norm{\D[2]{\psi}{x}}
			        \leq K_A\left(1 + \kappa(\varepsilon) \sqrt{|E|}\right) \left(\norm{D_B^2\psi} + \norm{\psi}\right).
			    \]
			    Obviously we can choose $\varepsilon_A$ such that $\kappa(\varepsilon_A) K_A \sqrt{|E|} = 1/2$, and then
			    \[
			    	\norm{\D[2]{\psi}{x}} \leq 2K_A\left(1 + \kappa(\varepsilon_A) \sqrt{|E|}\right)
			    	\left(\norm{D_B^2\psi} + \norm{\psi}\right)
			    	\eqqcolon \tilde{K}_A \left(\norm{D_B^2\psi} + \norm{\psi}\right).
			    \]

			    The proof can be finished by a compacity argument. Since we are only considering edge-constant vector potentials, each potential $A$ defines a point in $\mathbb{R}^{|E|}$. The subset $\mathcal{K}$ of $\mathbb{R}^{|E|}$ associated to the set of vector potentials satisfying $\max_{e \in E} A_e \leq r$ is a compact subset. Now, define $U_A = \{B \in \mathbb{R}^{|E|} \mid \max_{e \in E} |B_e - A_e| < \varepsilon_A\}$; the family $\{U_A\}_{A \in U}$ forms a covering of $\mathcal{K}$ and by compacity it admits a finite subcovering $\{U_{A_i}\}_i$. Taking the maximum of the associated constants,
			    \[
			    	K = \max_i \tilde{K}_{A_i},
			    \]
			    concludes the proof.
			\end{proof}

			\begin{thm}\label{thm:controllability-piecewise-smooth}
				Every simple quasi-$\delta$-type boundary control system is approximately controllable with $t \mapsto \chi(t) = A(t)x + b(t)$ (see \autoref{eq:form-of-chi}) piecewise $C^2$ and $b'(t), A'(t) \in L^2(G)$ such that for every $e,e' \in E$, $b'_e(t) = b'_{e'}(t)$ and for every vertex $v$, $\sum_{e \in E_v} \oveval{A'}_e(v,t) = 0$.
			\end{thm}
			\begin{proof}
				By \autoref{prop:quasi-periodic-BC-system-equivalence} the quasi-$\delta$-type boundary control system is controllable if and only if so is the magnetic controlled system with
				\begin{equation}\label{eq:original-hamiltonian}
					H(t) = - \left[\left(\Dx - iA(t)\right)^2 + b'(t) + A'(t) x \right]
				\end{equation}
				and $\delta$ vertex conditions. We will proof that this equivalent system is approximately controllable using \autoref{thm:chambrion-controllability}. The main problem to do this is the fact that $A(t)$ and $A'(t)$ are not independent, and to avoid this problem we need to proceed in two steps. First we define an auxiliary system to which \autoref{thm:chambrion-controllability} applies and then we use \autoref{thm:aprox-Hamiltonians-aprox-sol} to show that for any controls on the auxiliary system, its evolution is approximately the same as the evolution of the original system with some controls related to those on the auxiliary system.

				Let us start with the first step. Take a simple magnetic vector potential $a: G \to \R$ with associated magnetic Laplacian $-D^2 = -\left(\Dx - ia\right)^2$ and fix\footnote{If one is interested in having $b(t) \neq 0$, this proof can be extended straightforwardly to any $b(t)$ such that $b'(t)$ is edge and time independent.} $b(t) = 0$. We consider
				\begin{equation}\label{eq:def-betas}
					A'(t) = \{u(t)\beta_e\}_{e \in E},
				\end{equation}
				with $u$ a control and $\beta_e \in \mathbb{R}$, and define the auxiliary system with Hamiltonian
				\begin{equation}\label{eq:controllability-1}
					\tilde{H}(t) = -D^2 - u(t) \beta x,
				\end{equation}
				where $\beta: G \to \mathbb{R}$ denotes de edge-constant function with values $\beta_e$ ($e \in E$).

				It is easy to check that the assumptions made by Chambrion et al.\ are satisfied in our case: $H_0 = -D^2$ and $H_1 = \beta x$\footnote{Here $H_1$ must be understood edge by edge, that is, $H_1 \psi = \{\beta_e x_e \psi_e\}_{e \in E}$} are self-adjoint operators not depending on $t$, there exists an orthonormal basis of the Hilbert space $\H$ made of eigenfunctions of any magnetic Laplacian over $G$ provided that $G$ is compact \cite[Thm. 3.1.1]{berkolaiko2013introduction}, and $H_1$ is a bounded operator (since $G$ is compact) and thus $\dom H_1 = \H$.

				By \autoref{lemma:chambrion-thm-applies}, \autoref{thm:chambrion-controllability} can be applied (either to $\tilde{H}(t)$ or to a perturbed system with evolution as \emph{closed} as desired) and so the system is approximately controllable. Hence, for every initial state $\psi_0$, every target state $\psi_T$, every $\varepsilon > 0$ and every $c > 0$ there exists $T>0$ and $u(t): [0, T] \to (0, c)$ piecewise constant such that the evolution induced by $\tilde{H}(t)$, $\tilde{\psi}(t)$ satisfies $\tilde{\psi}(0) = \psi_0$ and $\norm{\tilde{\psi}(T) - \psi_T} < \varepsilon / 2$. Denote by $\tilde{U}(t,s)$ the unitary propagator associated to $\tilde{H}$ with that function $u(t)$.

				Now, choosing the vector potential from the original system \eqref{eq:original-hamiltonian} in such a way that its induced evolution is close enough to that of the auxiliary system, one guarantees that the evolved state reaches near the target state at time $T$. In order to do that, we split the time interval $[0, T]$ into $N$ pieces of length $\tau = T/N$, and for each of those subintervals define $A_k: [k\tau, (k+1)\tau) \to \mathbb{R}$ as
                \[
                    A_k(t) = \left\{a_e + \beta_e\int_{k\tau}^t u(s) \, ds\right\}_{e \in E} \in L^2(G),
                \]
                with $u(t)$ the piecewise control given by Chambrion et al.'s theorem. Taking
                \[
                    A(t) = \sum_{k = 0}^{N - 1} \chi_{[k \tau, (k+1) \tau)}(t) A_k(t),
                \]
                it is clear that $A'(t) = u(t)$ and $A \in \mathcal{C}_p(0,T)$. Also, by the mean value theorem,
                \begin{equation}\label{eq:controllability-3}
                    \norm{A - a}_\infty =
                    \adjustlimits\max_{k < N} \sup_{k\tau \leq t < (k+1)\tau}
                    \int_{k\tau}^t u(s) \,ds \leq c \tau
                \end{equation}
				For the moment, $\tau$ is arbitrary but later on we will need to choose it small enough. Note that if we take $\beta_e$ in \autoref{eq:def-betas} such that $\sum_{e \in E_v} \oveval{A'}_e(v,t) = 0$, $A(t)$ is a simple magnetic vector potential since so is $a$.

				Expanding the square on $D^2$ and having into account \autoref{lemma:magnetic-Laplacian-satisfy-thm5.3}, it is easy to check that, with the chosen $A(t)$,
				\begin{equation}\label{eq:controllability-4}
					H(t) = - \left( \Dx - iA(t)\right)^2 - A'(t) x
				\end{equation}
				fulfills the hypothesis of \autoref{thm:timedependent-linearCombination-Hamiltonian} in every interval $[k\tau, (k+1)\tau)$ and thus there exists a unitary propagator $U_k(t,s)$ describing the evolution induced by it for $t, s \in [k\tau, (k+1)\tau)$. For $t \in [k\tau, (k+1)\tau]$, $s \in [\ell \tau,
				(\ell + 1)\tau)$ with $\ell < k$ the unitary propagator is constructed multiplying them:
				\[
					U(t, s) = U_k(t, k\tau) U_{k-1}(k\tau, (k-1)\tau) \cdots U_\ell((\ell + 1)\tau, s).
				\]
				In what follows we omit the subscript on $U_k$ since the values of its arguments $t,s$ identify the index $k$ unambiguously.

				Finally, let $\{I_j\}$ with $I_j = [t_j, t_{j -1})$ be the coarser partition of $[0, T]$ which is a common refinement of both the partition $\{[k\tau, (k+1)\tau)\}_{k}$ and that given by the piecewise definition of $u(t)$. It is clear that the state of the system at time $T \in I_n$, assumed the evolution induced by $H(t)$ (defined in Equation \eqref{eq:controllability-4}) starting at $\psi_0$, can be written as
				\[
					\psi(T) = U(T, t_n) U(t_n, t_{n-1}) \cdots U(t_1, 0) \psi_0.
				\]
				And similarly for the state $\tilde{\psi}(T)$ if we assume evolution by $\tilde{H}$ defined in Equation \eqref{eq:controllability-1} (using the unitary propagator $\tilde{U}$ instead of $U$).

				It is straightforward to check that in every $I_j$ both Hamiltonians satisfy the hypothesis of \autoref{thm:aprox-Hamiltonians-aprox-sol}: the domain of magnetic Laplacians is fixed by $\delta$ vertex conditions independently of $t$, and the $\beta x$ operator is bounded. Both Hamiltonians satisfy the hypothesis of \autoref{thm:timedependent-linearCombination-Hamiltonian} (see \autoref{lemma:magnetic-Laplacian-satisfy-thm5.3}); and finally, $u(t)$ being $C^1(I_j)$ implies that the functions giving the time dependence of the Hamiltonians (after expanding the magnetic Laplacians) are also $C^1(I_j)$. Hence, for any $t, s \in I_j,$ and any $\varepsilon_2 > 0$ we can chose $\delta_1, \delta_2$ (defined in Eq. \eqref{eq:controllability-3} and Eq. \eqref{eq:controllability-2} respectively) as in \autoref{thm:aprox-Hamiltonians-aprox-sol} so that
				\[ \norm{\tilde{U}(t, s) \psi_0 - U(t, s) \psi_0} < \varepsilon_2. \]
				Hence, we have
				\[ \begin{alignedat}{2}
					\norm{\tilde{\psi}(T) - \psi(T)}
					&= \norm{\tilde{U}(T, t_n) \cdots \tilde{U}(t_1, s) \psi_0 -
							 U(T, t_n) \cdots U(t_1, s) \psi_0} \\
					&\leq \norm{\tilde{U}(T, t_n) \cdots \tilde{U}(t_2, t_1)U(t_1, s) \psi_0 - U(T, t_n) \cdots U(t_1, s) \psi_0} + \varepsilon_2 \\
					&\vdots \\
					&\leq (n+1) \varepsilon_2.
				\end{alignedat} \]
				Taking $\varepsilon_2 = \varepsilon / (2n+2)$, we have
				\[ \norm{\tilde{\psi}(T) - \psi(T)} \leq \frac{\varepsilon}{2}. \]

				Using that $\norm{\tilde{\psi}(T) - \psi_T} < \frac{\varepsilon}{2}$ we conclude
				\[ \norm{\psi(T) - \psi_T} < \varepsilon; \]
				that is, we have found controls $A(t)$ and $b(t)$ piecewise $C^2$ such that from any $\psi_0$ we can reach as close as we want to any $\psi_T$ and so the system is approximately controllable.
			\end{proof}
			Using \autoref{thm:controllability-piecewise-smooth} and an approximating argument similar to that in its proof, is easy to show that controls can also be smooth (not only piecewise smooth) functions of time.
			\begin{corol}\label{corol:controllability-smooth}
				Every simple quasi-$\delta$-type boundary control system is approximately controllable with $t \mapsto \chi(t) = A(t)x + b(t)$ smooth
				such that for every $e,e' \in E$, $b'_e(t) = b'_{e'}(t)$ and for every vertex $v$, $\sum_{e \in E_v} \oveval{A'}_e(v,t) = 0$.
			\end{corol}
			\begin{proof}
				Again, the quasi-$\delta$-type boundary control system is controllable if and only if so is the magnetic controlled system with
				\[ H(t) = - \left[\left(\Dx - iA(t)\right)^2 + b'(t) + A'(t) x \right] \]
				and $\delta$ vertex conditions.

				From \autoref{thm:controllability-piecewise-smooth}, for every initial state $\psi_0$, every target state $\psi_T$ and every $\varepsilon > 0$, we have piecewise $C^2$ controls $\tilde{A}(t)$ and $b(t) = 0$ such that the evolution $\tilde{\psi}(t)$ induced by
				\[ \tilde{H}(t) = - \left[\left(\Dx - i\tilde{A}(t)\right)^2 + b'(t) + \tilde{A}'(t) x \right], \]
				satisfies $\tilde{\psi}(0) = \psi_0$ and $\norm{\tilde{\psi}(T) - \psi_T} < \varepsilon / 2$. Denote by $\tilde{U}(t,s)$ the unitary propagator associated to $\tilde{H}(t)$.

				Using some well-known approximation result (see, for example, \cite[\S5.3]{evans_partial_1998}) one can find $A(t)$ smooth such that $\norm{A - \tilde{A}} < \delta_1$ and $\norm{A' - \tilde{A}'} < \delta_2$. Taking $\{I_j\}_j$ the partition of $[0, T]$ given by the subintervals on which $\tilde{A}(t)$ is $C^2$ and using the same argument as in the proof of \autoref{thm:controllability-piecewise-smooth}, one can use \autoref{thm:aprox-Hamiltonians-aprox-sol} to show that the evolution induced by $H(t)$ satisfies
				\[ \norm{\psi(T) - \psi_T} < \varepsilon.\]
			\end{proof}

			\subsection{Examples}
			Before we finish this section, let us consider some examples of quantum circuits to illustrate the power of the results presented above. First, let us conclude the example we have been developing in the previous sections.

			\begin{loopexample}
				In the part \ref{example:0-3} of this example we showed that every vector potential is simple for $G_0$, and therefore \autoref{thm:controllability-piecewise-smooth} and \autoref{corol:controllability-smooth} apply. Thus the quasi-$\delta$-type boundary control system associated with $G_0$ (see Fig. \ref{fig:example0}) is approximately controllable with control $\alpha$ smooth (or piecewise smooth).
			\end{loopexample}

			\begin{figure}[b]
				\centering
				\begin{subfigure}[b]{0.32\textwidth}
					\centering
					\includegraphics[width=\columnwidth]{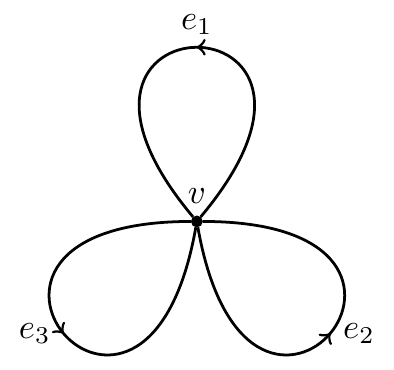}
					\caption{\autoref{example:1}: the bouquet $B_3$.}\label{fig:ex1}
				\end{subfigure}
				\begin{subfigure}[b]{0.32\textwidth}
					\centering
					\includegraphics[width=\columnwidth]{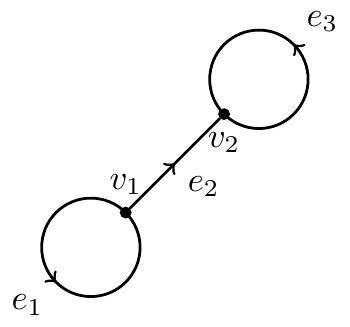}
					\caption{Graph $G_1$ in \autoref{example:2}.}\label{fig:ex2}
				\end{subfigure}
				\begin{subfigure}[b]{0.32\textwidth}
					\centering
					\includegraphics[width=\columnwidth]{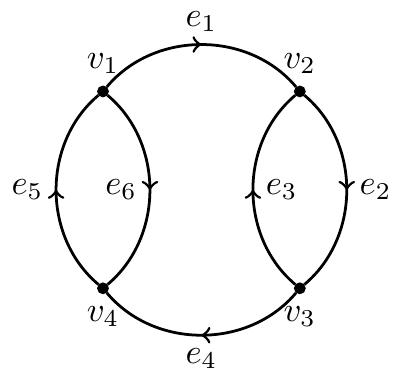}
					\caption{Graph $G_2$ in \autoref{example:3}.}\label{fig:ex3}
				\end{subfigure}
				\caption{Examples of quantum circuits that can be dealt with our study.} \label{fig:examples}
			\end{figure}
			\begin{example}\label{example:1}
				The first example we are going to study is the easiest one, the bouquet of order $n$, $B_n$, which is a graph with one vertex and $n$ loops (see \autoref{fig:ex1} for a graphical representation of $B_3$). In this case, it is clear that for any edge-constant magnetic vector potential $A$, we have
				\[
					\sum_{e \in E_v} \oveval{A}_e(v) = \sum_{e \in E_v^+} A_{e} - \sum_{e \in E_v^-} A_{e} = 0.
				\]
				Hence, any edge-constant vector potential is simple for $B_n$ and both \autoref{thm:controllability-piecewise-smooth} and \autoref{corol:controllability-smooth} apply:  the associated quasi-$\delta$-type boundary control system is approximately controllable with smooth (or piecewise smooth) associated vector potential whose derivative $A'(t) = \{A_e'(t)\}_{e \in E}$
				satisfies
				\[
					A'_e(t) = \beta_e A'_{e_0}(t) \qquad (\forall e \in E)
				\]
				for some reference edge $e_0 \in E$ and some real numbers $\beta_e$.
			\end{example}

			\begin{example}\label{example:2}
				In the previous example we studied graphs with just one vertex and we have shown that in such a case every edge-constant vector potential is simple, which leads to controllability with no constrain on $A'$. Let us consider now the graph $G_1$ represented in \autoref{fig:ex2}, consisting on two vertices joined by an edge, both of them with a loop. As we have seen in \autoref{example:1}, loops (i.e., edges starting and ending on the same vertex) do not contribute to the sum $\sum_{e \in E_v} \oveval{A}_e$ because the term corresponding to the edge \emph{entering} on the vertex cancels with the one \emph{leaving} the vertex. In this example, that makes that an edge-constant magnetic vector potential $A$ is simple for $G_1$ if and only if $A_{e_2} = 0$.

				In this case, applying \autoref{thm:controllability-piecewise-smooth} or \autoref{corol:controllability-smooth} one gets approximate controllability with smooth (or piecewise smooth) vector potentials such that $A_{e_2}(t) = 0$ for every $t$.
			\end{example}

			\begin{example}\label{example:3}
				For the last example, consider the graph $G_2$ showed in \autoref{fig:ex3}. In this case we study a graph with several vertices and no loops while having several closed paths. In this example, in order to be a simple edge-constant vector potential, $A$ needs to satisfy the following four equations:
				\[\begin{alignedat}{4}
					\text{Vertex $v_1$:} \quad & A_{e_5} = A_{e_1} + A_{e_6}; \qquad&
					\text{Vertex $v_2$:} \quad& A_{e_2} = A_{e_1} + A_{e_3}; \\
					\text{Vertex $v_3$:} \quad& A_{e_2} = A_{e_3} + A_{e_4}; \qquad&
					\text{Vertex $v_4$:} \quad& A_{e_5} = A_{e_4} + A_{e_6}.
				\end{alignedat}\]
				Thus, it is clear that $A$ has to be such that $A_{e_1} = A_{e_4}$. Then, by \autoref{thm:controllability-piecewise-smooth} or \autoref{corol:controllability-smooth} the system is approximately controllable with $A_{e_1}(t) = A_{e_4}(t)$ for all $t$.
			\end{example}

			Every generalized quasi-$\delta$-type boundary control system we have presented in this section is simple. We were not able to find an example which is not simple. That makes us conjecture that most graphs will be simple. However, even if we are not considering that in this work, not every magnetic vector potential is \emph{physically feasible} (i.e., can be obtained from a physical magnetic field). The condition for a vector potential to be \emph{physical} can be stated as follows: the sum of the vector potential over edges in a closed path needs to be equal to the flux of the associated magnetic field over that path. This imposes some conditions on the values of $A = \{A_e\}_{e \in E_v}$. We believe that this restriction is the one which is going to limit the set of graphs and vector potentials admissible.

	\section{Conclusions}
		We conclude this dissertation summarizing the ideas exposed. The main purpose of this work is to show that the formalism of Quantum Control at the Boundary (QCB) is indeed feasible and allows to control some quantum systems. To that end, we have been able to prove that some quantum circuits (that is, one-dimensional free quantum systems) can be approximately controlled interacting with them only through a change of the boundary conditions. Note that this way of interacting with the system is weaker than interacting applying external fields, and because of that it is expected to avoid decoherence (i.e., the lost of quantum effects in a composite system due to \emph{phase-shifts} induced by the interaction between the system and the environment). In addition, one-dimensional systems are the worst case scenario from the point of view of controllability since the space of controls, that is, the space of self-adjoint extensions, is parametrized by isometries acting on a finite dimensional space while in higher dimension the space of self-adjoint extensions is parametrized by isometries on infinite dimensional Hilbert spaces and is therefore much bigger. Moreover, the controllability result is obtained for a one-parameter subfamily. This supports the believe that QCB is feasible for general situations beyond the one considered here.

		Since the very idea of Quantum Control at the Boundary is quite novel and relies on properties intrinsic to unbounded operators, such that self-adjoint extensions and the dynamics induced by a family of self-adjoint extensions of the same symmetric operator, we had very few analytical results to rely on and we had to develop those analytical tools to prove the desired controllability results.

		In order to prove the sought approximate controllability, we first reviewed some known quantum graph theory \cite{berkolaiko2013introduction,kostrykin2003quantum} in sections \ref{sec:quantum-graphs}, \ref{sec:magnetic-laplacian}. Then, we were able to identify the proper analytical conditions in out setting in order to being able to apply some results in the book of M. Reed and B. Simon \cite[\S X.12]{reed_methods_1975}, leading us to results about the evolution of systems whose Hamiltonians depends on time through smooth functions, providing sufficient conditions for the existence of unitary propagators (\autoref{thm:timedependent-linearCombination-Hamiltonian}) and relating the evolution of two such systems with similar Hamiltonians (\autoref{thm:aprox-Hamiltonians-aprox-sol}). Finally we establish the equivalence between quasi-$\delta$-type boundary control systems (see \autoref{def:quasi-periodic-BCS}) and magnetic controlled quantum circuits, which allows to use the results of Chambrion et al.\ as an intermediate step to prove the approximate controllability.

		We want to remark that although we have proven controllability only for simple quasi-$\delta$-type boundary control systems, that does not imply that non-simple ones are not controllable. Non-simple quasi-$\delta$-type boundary control systems might be controllable but in order to deal with them one cannot use the theorems by Reed and Simon (\autoref{thm:reed-simon}) or by Chambrion et al.\ (\autoref{thm:chambrion-controllability}) because the domains of the associated Hamiltonians defining the dynamics are time-dependent. For that reason, to deal with non-simple systems one has to rely on results for non-constant domains such as that by J. Kisy\'nsky \cite{kisynski_sur_1964}. That seems the most natural step to generalize the results presented here.

		From this point, some lines of research may follow. Once controllability is achieved, it is interesting to address the optimal control problem as well as studying ways to apply QCB to quantum computation. On the other hand, it is also interesting to study how to generalize this results to the case of higher dimensions.

	\newpage
	\bibliographystyle{acm}
	\bibliography{bibliografia}
	\addcontentsline{toc}{section}{References}
\end{document}